\documentclass[10pt]{article}

\usepackage[utf8]{inputenc}
\usepackage[T1]{fontenc}
\usepackage[english]{babel}
\usepackage[a4paper,margin=1in]{geometry}

\usepackage{amssymb}
\usepackage{amsmath}
\usepackage{amsthm}
\usepackage{mathtools}

\usepackage[braket, qm]{qcircuit}

\usepackage{algorithm}
\usepackage{algpseudocode}

\usepackage{graphicx}
\usepackage{subcaption}
\usepackage{float}
\usepackage{tikz}
\usepackage{array}
\usepackage{makecell}
\usepackage{enumitem}

\newcommand{\C}{\mathbb{C}}
\newcommand{\R}{\mathbb{R}}

\newcommand{\Tr}{\operatorname{Tr}}

\newcommand{\Skew}{\operatorname{Skew}}
\newcommand{\Exp}{\operatorname{Exp}}
\newcommand{\Trt}{\operatorname{Trt}}

\newcommand{\grad}{\operatorname{grad}}
\newcommand{\poly}{\operatorname{poly}}
\newcommand{\proj}{\operatorname{Proj}}
\newcommand{\Hess}{\operatorname{Hess}}

\theoremstyle{plain}

\newtheorem{lemma}{Lemma}

\newtheorem{proposition}{Proposition}

\theoremstyle{definition}

\newtheorem{remark}{Remark}

\usepackage[colorlinks=true, allcolors=blue]{hyperref}

\usepackage[capitalize]{cleveref}

\newcommand{\email}[1]{\href{mailto:#1}{#1}}

\usepackage{lineno}

\usepackage{xcolor}
\definecolor{mygreen}{RGB}{205, 222, 194}

\begin{document}

  \title{Quantum circuit design from a retraction-based Riemannian optimization framework}

  \author{Zhijian Lai$^{1}$, Hantao Nie$^{1}$, Jiayuan Wu$^{2}$, Dong An$^{1}$\\
  \small $^{1}$Beijing International Center for Mathematical Research, Peking University, Beijing 100871, China\\
  \small $^{2}$Wharton Department of Statistics and Data Science, University of Pennsylvania, Philadelphia, PA 19104, USA\\
  \small Emails: \email{lai\_zhijian@pku.edu.cn}; \email{nht@pku.edu.cn}; \email{jyuanw@wharton.upenn.edu}; \email{dongan@pku.edu.cn}}
  \date{}

  \maketitle

  \begin{abstract}
  Designing quantum circuits for ground state preparation is a fundamental task in quantum information science. However, standard Variational Quantum Algorithms (VQAs) are often constrained by limited ansatz expressivity and difficult optimization landscapes. To address these issues, we adopt a geometric perspective, formulating the problem as the minimization of an energy cost function directly over the unitary group. We establish a retraction-based Riemannian optimization framework for this setting, ensuring that all algorithmic procedures are implementable on quantum hardware. Within this framework, we unify existing randomized gradient approaches under a Riemannian Random Subspace Gradient Projection (RRSGP) method. While recent geometric approaches have predominantly focused on such first-order gradient descent techniques, efficient second-order methods remain unexplored. To bridge this gap, we derive explicit expressions for the Riemannian Hessian and show that it can be estimated directly on quantum hardware via parameter-shift rules. Building on this, we propose the Riemannian Random Subspace Newton (RRSN) method, a scalable second-order algorithm that constructs a Newton system from measurement data. Numerical simulations indicate that RRSN achieves quadratic convergence, yielding high-precision ground states in significantly fewer iterations compared to both existing first-order approaches and standard VQA baselines. Ultimately, this work provides a systematic foundation for applying a broader class of efficient Riemannian algorithms to quantum circuit design.
  \end{abstract}

{\small
\par
}

\section{Introduction}

Designing quantum circuits to prepare specific quantum states lies at the heart of quantum information science.
This task, often framed as the ground state preparation for a given Hamiltonian, is fundamental to realizing the full potential of quantum computing across a broad range of domains.
Prominent applications include simulating electronic structures in quantum chemistry \cite{peruzzo2014variational,mcardle2020quantum,kandala2017hardware}, solving combinatorial optimization problems via the Quantum Approximate Optimization Algorithm (QAOA) \cite{farhi2014quantum,zhou2020quantum}, and enabling Quantum Machine Learning (QML) tasks \cite{biamonte2017quantum,schuld2015introduction}. In the Noisy Intermediate-Scale Quantum (NISQ) era \cite{preskill2018quantum}, improving the efficiency and practical feasibility of such circuit designs is particularly important.

Currently, the dominant paradigm for addressing these problems is the Variational Quantum Algorithm (VQA) framework \cite{cerezo2021variational,bharti2022noisy}, which relies on optimizing the Parameterized Quantum Circuits (PQCs) \cite{benedetti2019parameterized}.
Throughout the paper, we consider an $N$-qubit system with $p:=2^N$. Given an initial state $\psi_{0} = \lvert \psi_{0}\rangle\langle \psi_{0}\rvert$ and a target Hamiltonian $O$, the standard VQA approach aims to solve
\begin{equation}\label{pro-pqc}\tag{P1}
  \min_{\boldsymbol{\theta}\in\mathbb{R}^{m}}f (\boldsymbol{\theta})
  =\Tr\bigl\{O U (\boldsymbol{\theta}) \psi_{0} U (\boldsymbol{\theta})^{\dagger}\bigr\},
\end{equation}
where the structure of PQC (i.e., ansatz) $U(\boldsymbol{\theta})$ is fixed a priori, typically based on heuristic designs, such as the hardware-efficient ansatz \cite{kandala2017hardware}, unitary coupled-cluster ansatz \cite{peruzzo2014variational,romero2018strategies}, or Hamiltonian variational ansatz \cite{wecker2015progress,wiersema2020exploring}. The parameters $\boldsymbol{\theta}\in\mathbb{R}^{m}$ are optimized by classical methods. This paradigm is hardware-friendly as it keeps a constant circuit depth and uses well-established classical optimization routines.

However, the VQA framework faces some bottlenecks.
First, the fixed ansatz limits the expressivity of the model, because the true ground state may lie outside the reachable state space generated by $U(\boldsymbol{\theta})$ \cite{sim2019expressibility,holmes2022connecting}.
Second, the cost landscape of \eqref{pro-pqc} is often highly non-convex and riddled with local minima \cite{bittel2021training}. More critically, VQAs suffer from the barren plateau \cite{mcclean2018bp,larocca2025barren}, where gradients vanish exponentially with the number of qubits, making training effectively impossible.

In light of these challenges, there is growing interest in a geometric optimization perspective. The central idea is to lift the restriction of a fixed ansatz and instead optimize directly over the full space of quantum circuits.
Specifically, one formulates the task as a Riemannian optimization (also called manifold optimization) problem \cite{edelman1998geometry,absil2009optimization,boumal2023introduction,hu2020brief}:
\begin{equation}\label{pro-manifold}
  \tag{P2}
  \min_{U\in \mathrm{U} (p)}f (U)
  =\Tr\bigl\{O U \psi_{0} U^{\dagger}\bigr\},
\end{equation}
where the search space is the compact Lie group of $p\times p$ unitary matrices:
\begin{equation}
  \mathrm{U} (p)=\bigl\{ U\in\mathbb{C}^{p\times p} \big| U^{\dagger}U=I\bigr\}.
\end{equation}
This set, known as the unitary group\footnote{To simplify the description, we do not use the special unitary group $ \mathrm{SU} (p)$ in this paper. In the implementation, it suffices to remove the all-identity term from the Pauli words.}, is a Riemannian manifold of dimension $p^2=4^N$.
By treating the circuit $U$ itself as the variable on a manifold, one can use the rich geometric tools and algorithms from Riemannian optimization.

Recently, many works \cite{wiersema2023optimizing,pervez2025riemannian,malvetti2024randomized,magann2023randomized,mcmahon2025equating} have explored this geometric viewpoint.
These approaches typically start from the Riemannian gradient flow associated with \eqref{pro-manifold}, namely, $\frac{d}{d t} U(t)=-\operatorname{grad} f(U(t))$ with $U(t) \in \mathrm{U}(p),$ and then discretize it in time, leading to the Riemannian Gradient Descent (RGD) update $U_{k+1}=e^{-t_k\left[O, \psi_k\right]}  U_k$ with $U_0=I$, where the (Riemannian) gradient term $\left[O, \psi_k\right]$ is the commutator between $O$ and $\psi_k$, and $\psi_k=U_k \psi_0 U_k^{\dagger}$ is the current state.
However, a practical difficulty of RGD is the efficient implementation (or approximate simulation) of the additional gate $e^{-t_k\left[O, \psi_k\right]} $. Since this gate depends adaptively on $\psi_k$, realizing it typically requires repeatedly extracting information about the current state, leading to an exponential computational cost in terms of the number of iterations.
To address this issue, one often projects $\left[O, \psi_k\right]$ onto a chosen subset of Pauli words (e.g., the set of one- and two-local Pauli words) and implements the resulting evolution using a Trotter approximation \cite{wiersema2023optimizing}.
It has also been emphasized that randomized selection of Pauli words can substantially ensure the convergence behavior of such schemes \cite{pervez2025riemannian}.
Indeed, it can be viewed as a randomized variant of RGD on the manifold, for which convergence to the ground state from almost all initial states has been established \cite{malvetti2024randomized}.
A related variant that updates along a single randomly chosen Pauli word per iteration has also been studied \cite{malvetti2024randomized,magann2023randomized,mcmahon2025equating}.
The above approaches are collectively referred to as \textit{randomized RGD} methods.
On the other hand, the RGD update can also be interpreted through the lens of quantum imaginary-time evolution \cite{motta2020determining,mcmahon2025equating} and double-bracket quantum algorithms \cite{gluza2024double,suzuki2025double,suzuki2025grover}.

Although these studies have explicitly utilized Riemannian gradient descent, a comprehensive modern Riemannian optimization framework \cite{absil2009optimization,boumal2023introduction,hu2020brief} tailored for problem \eqref{pro-manifold} is still absent.
As a result, the literature has predominantly focused on first-order methods such as RGD, while second-order Riemannian methods remain largely unexplored regarding their implementability on quantum hardware.
Moreover, the convergence analyses do not directly use the established Riemannian optimization theory, which makes them problem-specific and hard for generalization to different choices in the algorithm design.

The modern Riemannian optimization framework \cite{absil2009optimization,boumal2023introduction,hu2020brief} is \textit{retraction-based}, where retractions \cite{adler2002newton} are the most general class of mappings that map tangent vectors back onto the manifold while still guaranteeing convergence.
In this paper, we develop a complete retraction-based framework for problem \eqref{pro-manifold}, in which every algorithmic component is implementable on quantum hardware.
The main contributions of this work can be summarized as follows.

\begin{enumerate}
  \item \textbf{First-order methods.} We formalize the Trotter approximation \cite{trotter1959product,suzuki1991general} as a retraction on unitary group, thereby bridging the algorithmic design with standard optimization methodology.
        We then interpret the randomized RGD methods introduced above \cite{pervez2025riemannian,malvetti2024randomized,magann2023randomized,mcmahon2025equating} as a \emph{Riemannian Random Subspace Gradient Projection} (RRSGP) method: at each iteration, the full Riemannian gradient is randomly orthogonally projected onto a $d=\poly(N)$-dimensional subspace of the current tangent space, and the update is taken along the resulting projected gradient.
        Furthermore, we introduce an exact line-search technique to determine the optimal step size, which is equivalent to solving a single-parameter VQA subproblem.

  \item \textbf{Second-order methods.} We derive an explicit expression for the Riemannian Hessian (i.e., the manifold analogue of the Euclidean Hessian) of the cost function in \eqref{pro-manifold}. Crucially, we demonstrate that the Riemannian Hessian can be estimated directly on quantum hardware via quantum measurements using parameter-shift rules \cite{mari2021estimating,schuld2019evaluating}, just as existing studies \cite{wiersema2023optimizing,pervez2025riemannian,mcmahon2025equating} estimate the Riemannian gradient through quantum measurements.
        Building on this, we propose a fully quantum-implementable Riemannian Newton algorithm.
        To ensure scalability, we adopt a randomized strategy analogous to RRSGP and introduce the \textit{Riemannian Random Subspace Newton} (RRSN) method. RRSN constructs a $d\times d$ Newton system from measurement data, stabilized by regularization techniques.
        Our numerical results demonstrate that RRSN achieves a quadratic convergence rate, attaining the ground state in significantly fewer iterations compared to randomized RGD and VQA baselines.
        We show that a shallow VQA warm start effectively combines the hardware efficiency of PQCs with the high-accuracy refinement of Riemannian optimization.

\end{enumerate}

\paragraph{Organization}
In \cref{sec-1}, we review the first-order geometry of the unitary group and introduce the quantum-implementable retractions. We then develop the first-order algorithm framework, specifically the Riemannian Random Subspace Gradient Projection (RRSGP) method. \cref{sec-2} focuses on the second-order geometry and algorithms: we derive the explicit form of the Riemannian Hessian, show how to estimate it via measurements, and propose the Riemannian Random Subspace Newton (RRSN) method. \cref{sec-experiments} presents numerical simulations on the Heisenberg XXZ model to evaluate the performance of our algorithms.
We conclude in \cref{sec-discussion}.

\section{First-order geometry and algorithms}\label{sec-1}

This section provides a minimal review of the first-order geometry of the unitary group, followed by an introduction to several first-order Riemannian algorithms for problem \eqref{pro-manifold} that can be implemented on quantum hardware. For a complete and rigorous derivation of geometry, we refer the reader to the monographs \cite{absil2009optimization,boumal2023introduction}.
Second-order geometry and algorithms will be presented in the next section. Throughout the paper, we set $p=2^N$ for an $N$-qubit system.

\subsection{First-order geometry}

First, we review the fundamental geometric concepts and tools required for optimization on the unitary group.

\subsubsection{Tangent space and orthogonal projection}

The tangent space of a manifold at a given point consists of all directions in which one can move from that point. For the unitary group $\mathrm{U}(p)$, this space admits an explicit description in terms of skew-Hermitian matrices. For any $U\in \mathrm{U} (p)$, the tangent space at $U$ is given by
\begin{equation}
  T_U
  =\{ \Omega U:\Omega ^\dagger=-\Omega\}=\mathfrak{u} (p)U,
\end{equation}
where $\mathfrak{u}(p)=\left\{\Omega \in \mathbb{C}^{p \times p}: \Omega^{\dagger}=-\Omega\right\}$ is the Lie algebra of $\mathrm{U}(p)$. This Lie algebra is a real vector space consisting of all skew-Hermitian $p \times p$ matrices and is closed under the Lie bracket $[X, Y]=X Y-Y X$. In particular, at the identity $I\in \mathrm{U} (p)$, one has $T_I= \mathfrak{u}(p)$; since $\operatorname{dim} \mathfrak{u}(p)=p^2$, the manifold $\mathrm{U}(p)$ has dimension $p^2$.

We equip the ambient space $\mathbb{C}^{p \times p}\supseteq \mathrm{U}(p)$ with the real Frobenius inner product $\langle A, B \rangle = \Re \Tr (A^\dagger B).$ With respect to this inner product, the normal space at $U \in \mathrm{U} (p)$ is given by $N_U   = \bigl\{ N\in\C^{p\times p}: \langle N, X \rangle =0, \forall X\in T_U \bigr\} =  \{ HU:  H ^\dagger=H  \} =  \mathcal{H} (p) U,$ where $\mathcal{H}(p)$ denotes the space of Hermitian $p \times p$ matrices. Using the orthogonal decomposition $\mathbb{C}^{p \times p}=T_U \oplus N_U$, one obtains the orthogonal projection of any $Z \in \mathbb{C}^{p \times p}$ onto the tangent space $T_U$:
\begin{equation}\label{eq-projection}
  \mathcal{P}_U (Z): =\operatorname{Skew} (Z U^{\dagger}) U, \text { with } \operatorname{Skew} (A): =\tfrac{1}{2} (A-A^{\dagger}) .
\end{equation}

\begin{figure}
  \centering
  \includegraphics[width=1\linewidth]{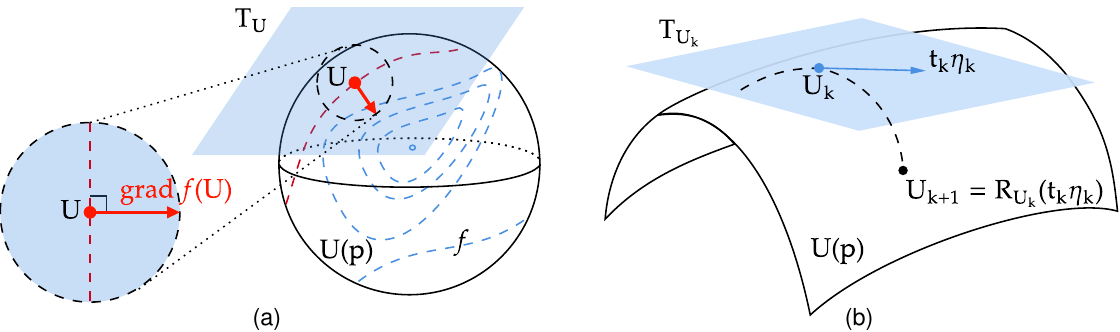}
  \caption{The geometric illustration of Riemannian gradient and Riemannian algorithm iterations.
    (a) For the sake of demonstration, we imagine the curved space $\mathrm{U} (p)$ as a sphere. Given a point $U$ on $\mathrm{U} (p)$, the tangent space $T_U$ is similar to the tangent plane at that point. We plot the contours (dash lines) of the cost function $f$ on $\mathrm{U} (p)$, where the Riemannian gradient $\grad f (U)$ is a special vector in the tangent plane. This vector is orthogonal to the contour and points in the direction of the fastest increase of $f$.
    (b) The iteration process of a typical Riemannian optimization algorithm $U_{k+1}=\mathrm{R}_{U_k}\left(t_{k} \eta_{k}\right)$, which is an extension of  iteration $x_{k+1} = x_k+t_k \eta_k$ in Euclidean space.}
  \label{fig:shiyi}
\end{figure}

\subsubsection{Riemannian gradient}

We regard $\mathrm{U}(p)$ as a Riemannian submanifold of $\mathbb{C}^{p \times p}$ by equipping each tangent space $T_U$ with the \textit{Riemannian metric} (some inner product) $ \langle X, Y \rangle_{U} \equiv \langle X, Y \rangle =  \langle \widetilde{X}, \widetilde{Y} \rangle = \Tr (\widetilde{X}^ \dagger \widetilde{Y})$ for all $X, Y \in T_U$, where $\widetilde{X}, \widetilde{Y} \in \mathfrak{u}(p)$ are the left skew-Hermitian representations, defined through $X=\widetilde{X} U$, $Y=\widetilde{Y} U$.
With this metric, the Riemannian gradient of any smooth function $f:\mathrm{U}(p) \subseteq \mathbb{C}^{p \times p} \to \mathbb{R}$ can be computed by (see \cite[Proposition 3.61]{boumal2023introduction})
\begin{equation}\label{eq-1530}
  \grad f (U)
  = \mathcal P_U\bigl (\nabla f (U)\bigr)
  \in T_U,
\end{equation}
where $\nabla f (U) \in \mathbb{C}^{p \times p} $ is the usual gradient. At a point $U$, the Riemannian gradient $\grad f (U)$ is the tangent vector of steepest ascent of $f$ and thus forms a basic component of Riemannian optimization algorithms; see (a) in \Cref{fig:shiyi}.

For our cost $f (U)= \Tr\left\{O U \psi_0 U^{\dagger}\right\}$ in \eqref{pro-manifold}, the Euclidean gradient is $\nabla f (U) = 2 O U \psi_0 .$ In general, when $A$ and $B$ are (skew-)Hermitian, then
\begin{equation}\label{eq-iden-skew}
  \Skew (2AB) =[A, B]
\end{equation}
holds. Applying the above with $A := O$ and $B := U \psi_0 U^\dagger$, and invoking \cref{eq-1530}, we obtain
\begin{equation}\label{eq-grad}
  \grad f (U)
  =\Skew (2 O U \psi_0 U^\dagger)U
  =[O, U \psi_0 U^\dagger]U \in T_U,
\end{equation}
where the commutator $\widetilde{\grad} f(U):=[O, U \psi_0 U^\dagger]\in \mathfrak{u} (p)$ serves as the left skew-Hermitian representation of $\grad f (U)$.

\subsubsection{Quantum-implementable retractions}

\textit{Retractions} \cite{adler2002newton,absil2009optimization} provide a principled way to move along tangent directions while ensuring that iterates remain on the manifold, and they constitute a central step of Riemannian optimization theory.
Define the tangent bundle $T\mathrm{U}(p) := \left\{(U, \eta) : U \in \mathrm{U}(p), \eta \in T_U \right\}.$
A retraction on $\mathrm{U}(p)$ is a smooth map $\mathrm{R}: T\mathrm{U}(p) \rightarrow \mathrm{U}(p),$ $(U, \eta) \mapsto \mathrm{R}_U(\eta),$ such that for any $U \in \mathrm{U}(p)$ and $ \eta \in T_U $, the induced curve $t \mapsto \gamma(t) := \mathrm{R}_U(t \eta)$ satisfies
\begin{equation}\label{eq-defn-retraction}
  \gamma (0)=U \text { and } \dot{\gamma} (0)=\eta .
\end{equation}
The above conditions are the key to ensuring convergence of the Riemannian algorithm, rather than the specific form of the retraction. We denote by $\mathrm{R}_U$ the restriction of R to $T_U$.
Once a retraction is fixed, the next iterate $U_{k+1}$ is obtained by selecting a tangent direction $\eta_k \in T_{U_k}$ at the current iterate $U_k$ and moving along the curve $t \mapsto \mathrm{R}_{U_k}\left (t \eta_k\right)$ with an appropriate step size $t_k$; see (b) in \Cref{fig:shiyi}.
This mechanism forms the core update rule of Riemannian optimization, whose details will be presented in \cref{sec-first-algo}.
Below, we discuss two retractions on $\mathrm{U}(p)$ that are \textit{implementable} on quantum hardware.

\paragraph{Riemannian Exponential map}
For any $U \in \mathrm{U}(p)$, the (Riemannian) Exponential map $\Exp_{U}: T_{U} \rightarrow \mathrm{U} (p)$ is given by
\begin{equation}\label{eq-Exp}
  \Exp_{U} (\eta)
  =\exp \bigl (\widetilde{\eta}\bigr) U,
\end{equation}
where we often use the symbol tilde to denote the left skew-Hermitian representation of the tangent vector in $T_{U}$, i.e., $\eta=\tilde{\eta} U$ with $\tilde{\eta} \in \mathfrak{u}(p)$. Consider the curve
\begin{equation}\label{eq-geodesic}
  \gamma (t): = \Exp_U (t \eta ) = \exp (t \widetilde{\eta}) U.
\end{equation}
It is straightforward to verify that $\gamma (0) = U$, and $\dot\gamma (0)= \left. \widetilde{\eta} \exp (t \widetilde{\eta})U \right|_{t=0}= \eta.$ Therefore, the Exponential map serves as a valid retraction.

The Exponential map is a fundamental concept in differential geometry \cite{lee2003smooth}; the curve it induces in \cref{eq-geodesic} provides the intrinsic analogue of a straight line on a manifold, namely, a geodesic.
On the other hand, we use $\exp (\cdot)$ or $e^{(\cdot)}$ for the ordinary matrix exponential, which should be distinguished from Exp.
In general, the Exponential map on a manifold does not always involve a matrix exponential. For instance, on the trivial manifold $\mathbb{R}^n$, the Exponential map takes the simple form $\Exp_x(\eta) = x + \eta$ for $x \in \mathbb{R}^n$ and $\eta \in T_x \equiv \mathbb{R}^n$.

In fact, the Exponential map is seldom used in Riemannian optimization due to its high computational cost on classical computers. For our manifold $\mathrm{U}(p)$, more efficient alternatives (such as the Cayley transform retraction \cite{wen2013feasible}, QR-based retraction \cite{sato2019cholesky,absil2009optimization}, or polar decomposition \cite{absil2012projection,absil2009optimization}) are typically preferred. However, the curve in \cref{eq-geodesic} is particularly well suited for implementing iterative updates in quantum circuits because $\exp(t\widetilde{\eta})=\exp(-i t\cdot i\widetilde{\eta})$ constitutes a standard quantum gate, with the step size $t$ serving as the rotation angle and $i\widetilde{\eta}$ acting as the Hermitian generator.

\paragraph{Trotter retraction}
The well-known first-order Trotter approximation \cite{trotter1959product,lloyd1996universal} also induces a valid retraction.
Let $\mathcal{P}^N=\{P^j\}_{j=1}^{4^N}$ denote the $N$-qubit Pauli words, where $P^j=\bigotimes_{\ell=1}^N \sigma_\ell^j$ with $\sigma_\ell^j\in\{I, X, Y, Z\}$. Then $i\mathcal{P}^N=\{iP^j\}$ is an orthogonal (non-normalized) basis of Lie algebra $\mathfrak{u} (p)$.
For any $U\in \mathrm{U} (p)$ and $\eta\in T_U$, write $\eta=\tilde{\eta}U$ with $\tilde{\eta}\in\mathfrak{u} (p)$, and expand $\tilde{\eta}$ in the Pauli basis:
\begin{equation}\label{eq-pauli-expansion}
  \tilde{\eta}=\sum_{j=1}^{4^N} i \omega^j P^j,
  \qquad
  \omega^j= \frac{\langle \tilde{\eta} , iP^j\rangle}{\langle iP^j , iP^j\rangle} =\frac{-i}{2^N}\Tr  (\tilde{\eta}P^j).
\end{equation}
The \textit{Trotter retraction} $\operatorname{Trt}_U: T_U \to \mathrm{U} (p)$ is defined by
\begin{equation}\label{eq-trt}
  \Trt_U(\eta) = \left( \prod_{j=1}^{4^N} \exp\left( i \omega^j P^j \right) \right) U.
\end{equation}
Consider the curve
\begin{equation}
  \gamma(t) := \Trt_U(t \eta) = \left( \prod_{j=1}^{4^N} \exp\left(i \omega^j P^j t\right) \right)U \equiv A(t) U.
\end{equation}

It is immediate that $\gamma(0) = U$. Applying the product rule yields 
\begin{equation}
  \dot{A}(t) = \sum_{k=1}^{4^N} \left( \prod_{j<k} \exp(i \omega^j P^j t) \right) \left( i \omega^k P^k \exp(i \omega^k P^k t) \right) \left( \prod_{j>k} \exp(i \omega^j P^j t) \right).
\end{equation}
Evaluating at $t = 0$, we obtain $\dot{\gamma}(0) = \dot{A}(0) U = \left( \sum_{k=1}^{4^N} i \omega^k P^k \right) U = \eta.$
Hence, $\Trt_U$ is a valid retraction. Obviously, this retraction stems from the first-order Trotter approximation:
\begin{equation}
  \exp\left(t\widetilde{\eta}\right) = \exp\left( \sum_{j=1}^{4^N} i\omega^j P^j t\right) = \prod_{j=1}^{4^N} \exp\left(i \omega^j P^j t\right) + O(t^2).
\end{equation}
The same idea extends naturally to the second-order \cite{strang1968construction} and higher-order Trotter schemes \cite{suzuki1991general}, producing Trotter-like retractions that more closely approximate the geodesic $t \mapsto \exp (t \tilde{\eta})$.
However, this refinement is unnecessary, since the first-order construction in \cref{eq-trt} already provides sufficient convergent properties, both in our theoretical analysis and in numerical experiments for Riemannian optimization.

We refer to the two retractions defined above, $\Exp$ and $\Trt$, as \textit{quantum-implementable retractions} because they both take the form of a standard quantum gate.
In contrast, many other retractions on the unitary group do not admit such physical implementations.
For example, the Cayley transform retraction \cite{wen2013feasible} requires matrix inversion, while the QR-based retraction \cite{sato2019cholesky,absil2009optimization} and polar decomposition \cite{absil2012projection,absil2009optimization} retraction rely on matrix factorizations. These procedures cannot be executed directly on quantum hardware.
One may say that the development of quantum computing has brought renewed attention to certain mathematical objects that were previously overlooked by classical computation.

\subsection{First-order algorithms}\label{sec-first-algo}

In this subsection, we first review the general framework of Riemannian optimization and then discuss the simplest first-order method, the Riemannian Gradient Descent (RGD). However, implementing RGD on quantum hardware can suffer from the curse of dimensionality. To address this, we consider a variant that projects the Riemannian gradient onto a randomly low-dimensional subspace, namely the Riemannian Random Subspace Gradient Projection (RRSGP). Finally, we also consider employing exact line-search to accelerate convergence. Our guiding principle is to ensure full consistency between circuit implementation and optimization framework, while using optimization techniques to enhance circuit design.

\subsubsection{General Riemannian algorithm framework}

In the Euclidean case, minimizing $f: \mathbb{R}^n \rightarrow \mathbb{R}$ often amounts to generating a sequence $\left\{x_k\right\}$ via updates of the form $x_{k+1}=x_k+t_k \eta_k,$ where $\eta_k \in  \mathbb{R}^n $ is a search direction (e.g., negative gradient or Newton direction) and $t_k>0$ is a step size, chosen as a fixed constant \cite{nesterov_introductory_2004} or by a line-search strategy such as Armijo backtracking \cite{nocedal1999numerical}, exact line-search \cite{nocedal1999numerical}, or Barzilai-Borwein rule \cite{raydan1993barzilai}.

On the $\mathrm{U}(p)$, this idea is adapted by replacing straight-line steps with moves along tangent directions followed by a retraction back onto the manifold. Given an initial point $U_0 \in \mathrm{U}(p)$, a Riemannian algorithm generates $\left\{U_k\right\} \subseteq \mathrm{U}(p)$ by
\begin{equation}\label{general-update}
  U_{k+1} = \mathrm{R}_{U_k}(t_k \eta_k), \quad k=0,1, \ldots,
\end{equation}
where $\eta_k \in T_{U_k}$ is a tangent vector in the current tangent space (e.g., the negative Riemannian gradient in \cref{eq-grad}, or the Riemannian Newton direction to be introduced in \cref{sec-2}), and $\mathrm{R}_{U_k}: T_{U_k} \rightarrow \mathrm{U}(p)$ is a retraction at $U_k$, such as the Exponential map $\Exp$ or the Trotter retraction $\Trt$ introduced previously; see (b) in \Cref{fig:shiyi}.
By choosing $\eta_k$ and $t_k$ in analogy with their Euclidean counterparts (e.g., Riemannian gradient descent \cite{edelman1998geometry,boumal2019global}, Riemannian conjugate gradient \cite{sato2022riemannian,abrudan2008efficient}, Riemannian Newton \cite{absil2009optimization,fernandes2017superlinear}, or Riemannian Barzilai-Borwein \cite{iannazzo2018riemannian}), one obtains Riemannian algorithms that inherit the corresponding convergence guarantees from the Euclidean setting. \cref{general-update} also serves as a prerequisite for the second-order algorithms in  \cref{sec-2} later.

\subsubsection{Riemannian gradient descent (RGD)}

We begin by considering the simplest method, namely Riemannian Gradient Descent (RGD) method with the Exponential map $\mathrm{R}= \Exp$ for problem \eqref{pro-manifold}. Its update rule is
\begin{equation}\label{eq-full}
  U_{k+1}
  = \Exp_{U_{k}} \left(- t_k \grad   f (U_k) \right)
  =\exp \left(t_k[\psi_k, O]\right) U_k  ,\quad k= 0,1,\dots,
\end{equation}
where step size $t_k>0$, $\psi_k: = U_k \psi_0 U_k^\dagger$ denotes the intermediate quantum state at iteration $k$, and $-\widetilde{\grad} f\left(U_k\right)=\left[\psi_k, O\right]\in \mathfrak{u}(p)$. Setting the initial gate to $U_0 = I$ (as we do throughout this paper), the resulting quantum state $\lvert \psi_T \rangle = U_T \lvert \psi_0 \rangle$ after $T$ updates is given by
\begin{equation}
  \left|\psi_T\right\rangle
  =
  \exp \left (t_{T-1}\left[\psi_{T-1}, O\right]\right) \cdots
  \exp \left (t_1\left[\psi_1, O\right]\right) \exp \left (t_0\left[\psi_0, O\right]\right)
  \left|\psi_0\right\rangle .
\end{equation}
The overall plan operates as follows: information obtained from the state generated by the current circuit is used to select additional gates that, when appended, drive the circuit output closer to the ground state. This procedure is applied iteratively, yielding a progressively refined circuit.
In contrast to traditional VQAs, whose ansatz is fixed in advance, the circuit architecture here evolves dynamically during the optimization process. Such dynamical architectures are referred to as adaptive circuits, for example, AdaptVQE \cite{grimsley2019adaptive,tang2021qubit}.

The key challenge now is to implement the update gates $\exp\bigl(t_k [\psi_k, O]\bigr)$ on circuits.
Although one can realize this via various Hamiltonian simulation techniques, here we instead employ the Trotter retraction $\mathrm{R} = \Trt$, whose quantum implementation is simpler than that of $\mathrm{R} = \Exp$.
From the standpoint of convergence, employing the Trotter retraction still provides the guarantees we require, since it satisfies conditions \cref{eq-defn-retraction}. Consequently, using the Trotter retraction, the RGD update takes the form
\begin{equation}\label{eq-update-Trt}
  U_{k+1}
  = \Trt_{U_{k}}\bigl(-t_k \grad f (U_k)\bigr)
  = \left( \prod_{j=1}^{4^N} \exp\left(i \omega^j _kP^j t_k\right) \right) U_k  ,\quad k= 0,1,\dots,
\end{equation}
where $-\widetilde{\grad} f (U_k)=[\psi_k, O]\in \mathfrak{u}(p)$ admits the Pauli expansion (see \cref{eq-pauli-expansion}):
\begin{equation}\label{eq-1823}
  [\psi_k, O] = \sum_{j=1}^{4^N} i \omega_k^j P^j, \quad \omega_k^j: = \tfrac{-i}{2^N} \Tr\left\{[\psi_k, O] P^j \right\} = -\tfrac{1}{2^N} i\Tr\left\{ \psi_k[O, P^j] \right\} \in \mathbb{R},
\end{equation}
where the last equality follows from $\Tr ([A, B] C) = \Tr (A[B, C])$.

In fact, we can determine the value of $\omega_k^j$ by measurements on the current state $\psi_k$.
To this end, we require the following lemma; its proof is provided in \ref{app-proofs}. (For the moment, we ignore the exponential term $4^N$ and focus on the exact algorithmic formulation.)

\begin{lemma}[Gradient coefficient estimation]\label{lem-pqc-1}
  Given any pure density operator $\psi$, any Hamiltonian $O$, and any Hermitian generator $P$, consider the expectation value on a variational circuit w.r.t. a single real parameter $x$, i.e.,
  \begin{equation}
    g (x)=\Tr \{ O U (x) \psi U (x)^{\dagger}\},
  \end{equation}
  where $U (x)=e^{i x P / 2}.$ Then, $g^{\prime}(0)=\tfrac{i}{2} \Tr\{\psi[O, P]\}.$
  Moreover, if  $P^2=I$, then by the parameter-shift rule \cite[Eq. (9)]{mari2021estimating}, one has
  \begin{equation}
    g^{\prime}(0)=\frac{1}{2}\left[g\left(\frac{\pi}{2}\right)-g\left(-\frac{\pi}{2}\right)\right].
  \end{equation}
\end{lemma}

To estimate the coefficients $\omega_k^j$ in \cref{eq-update-Trt}, we consider the following PQC cost function:
\begin{equation}
  g_k^j (x) :=  \Tr \{O U (x)\psi_k U (x)^\dagger \}, \quad U (x) = e^{i x P^j / 2}.
\end{equation}
Since $P^j$ is a Pauli word, by \cref{lem-pqc-1} and \cref{eq-1823}, we obtain
\begin{equation}\label{eq-1108}
  \omega_k^j = \frac{1}{2^N}\left[g_k^j\left(-\frac{\pi}{2} \right) - g_k^j\left(\frac{\pi}{2} \right)\right].
\end{equation}
Thus, $\omega_k^j$ can be directly estimated on a quantum device.
For any Pauli word $P^j$ and any real evolution time $\omega_k^j t_k$, the corresponding unitary $\exp(i P^j \cdot (\omega_k^j  t_k ))$ can be synthesized using standard Hamiltonian simulation techniques. In this work, we adopt such Pauli-evolution gates as fundamental building blocks.

Consequently, ignoring the exponential scaling $4^N$, the standard RGD with the Trotter retraction can be implemented entirely using quantum circuits, achieving a seamless unification of quantum implementation and optimization framework.

\subsubsection{Riemannian random subspace gradient projection (RRSGP)}\label{sec-rrsgp}

Since the Trotter update in \cref{eq-update-Trt} demands exponentially many gates per iteration, which is clearly impractical, a simple remedy is to choose a much smaller subset of $d$ Pauli words at random from the full set of $4^N$. The update then becomes
\begin{equation}\label{eq-trt-random}
  U_{k+1}
  = \left( \prod_{j\in S_k} \exp\left(i \omega^j _k P^j t_k\right) \right) U_k,\quad k=0,1, \ldots,
\end{equation}
where $S_k\subseteq\{1,2,\dots,4^N\}$ is a random subset of size $d=\mathrm{poly}(N)$. Each coefficient $\omega^j_k$ is computed as \cref{eq-1108}, and this procedure adds only $d$ gates per update. In fact, we are not introducing a new retraction here, but rather a Riemannian analogue of Random Subspace Gradient Projection (RSGP) method.

Let us review RSGP in the Euclidean setting. To solve $\min_{x\in\mathbb{R}^n} f(x)$ with very large $n$, one projects the full gradient onto a random $d$-dimensional subspace $\mathcal{U}_k \subseteq \mathbb{R}^n$ ($d \ll n$) and then takes a descent step within that subspace. Concretely, the RSGP update is
\begin{equation}\label{eq-ssgp}
  x_{k+1}
  = x_{k}
  + t_k \, \Pi_k \bigl(-\nabla f(x_k)\bigr),
\end{equation}
where $\Pi_k \in \mathbb{R}^{n\times n}$ is the orthogonal projector onto $\mathcal{U}_k$. Equivalently, one often writes $\Pi_k = P_k P_k^{\top}$, where $P _k\in\mathbb{R}^{n\times d}$ is a random matrix satisfying $P _k^{ \top}P _k =I_d$ and $\mathbb{E}\bigl[P _kP _k^{ \top}\bigr] =  \frac{d}{n}  I_n$.
This RSGP approach preserves convergence guarantees in expectation, comparable to those of full gradient descent \cite{kozak2019stochastic,kozak2021stochastic}.

In the Riemannian context, we will orthogonally project the full Riemannian gradient $\grad f \left(U_k\right) \in T_{U_k}$ onto a random $d$-dimensional subspace of the current tangent space $T_{U_k}$. Since each tangent space $T_{U_k}=\mathfrak{u}(p)U_k$ is isomorphic to the Lie algebra $\mathfrak{u}(p)$, it suffices to project the skew-Hermitian part $\widetilde{\grad} f \left(U_k\right)=\sum_{j=1}^{4^N} i \omega_k^j P^j \in \mathfrak{u}(p)$ onto a random $d$-dimensional subspace of $\mathfrak{u}(p)$.
Concretely, we sample $d$ Pauli words from the full basis $\{i P^j\}$, and let the selected index set be $S_k \subseteq \{1,2, \ldots, 4^N\}$. We use $\mathcal{S}_k \subseteq T_{U_k}$ to denote the subspace spanned by these selected directions. The projected gradient is then
\begin{equation}\label{eq-zeta}
  \zeta_k := \proj_{\mathcal{S}_k}(-\grad f \left(U_k\right))=\left(\sum_{j\in S_k} i \omega_k^j P^j \right) U_k \in \mathcal{S}_k.
\end{equation}
Substituting $\zeta_k$ into the Trotter retraction $U_{k+1}=\Trt_{U_k}\left(t_k \zeta_k\right)$ exactly recovers the update in \cref{eq-trt-random}. The complete RRSGP procedure is summarized in \cref{alg-RRSGP}.

\begin{algorithm}[H]
  \caption{Riemannian Random Subspace Gradient Projection (RRSGP) Method for \eqref{pro-manifold}}
  \label{alg-RRSGP}
  \begin{algorithmic}[1]
    \Require
    $N$-qubits system, Hamiltonian $O$, initial state $\psi_0$,  subspace dimension $d = \poly(N)$. Pauli set $\mathcal{P}^N=\{P^j\}_{j=1}^{4^N}$ with $P^j=\bigotimes_{\ell=1}^N \sigma_{\ell}^j$, $\sigma_{\ell}^j\in\{I,X,Y,Z\}$.

    \For{$k = 0, 1, \ldots$}

    \State \textbf{Step 1: Subspace Sampling}
    \State Uniformly sample index set $S_k \subseteq \{1, \dots, 4^N\}$ of size $d$.

    \vspace{0.4em}
    \State \textbf{Step 2: Gradient Estimation on Subspace} \Comment{$-\widetilde{\grad} f\left(U_k\right)$}
    \For{$j \in S_k$}
    \State Define $g_k^j(x)=\Tr \{O e^{i x P^j / 2} \psi_k e^{-i x P^j / 2}\}$.
    \State Set $\omega_k^j = \frac{1}{2^N} [ g_k^j(-\frac{\pi}{2}) - g_k^j(\frac{\pi}{2})]$.  \Comment{By parameter-shift.}
    \EndFor

    \vspace{0.4em}
    \State \textbf{Step 3: Step Size Selection}
    \State Choose a step size $t_k$. \Comment{Fixed step (e.g., $0.1$) or exact line-search.}
    \State Append update gates
    $|\psi_{k+1}\rangle \leftarrow  \prod_{j \in S_k} \exp (i \omega_k^j  P^j t_k)   |\psi_{k} \rangle $.
    \State Set $k \gets k + 1$.
    \EndFor
  \end{algorithmic}
\end{algorithm}

The uniform random selection of a new subspace at each iteration is critical to preserve convergence of RRSGP.
For instance, the method in \cite{wiersema2023optimizing} employs a structure similar to \cref{alg-RRSGP} but restricts optimization to a fixed subspace (e.g., the one-local or two-local Pauli set); consequently, the theoretical convergence guarantees no longer hold.
More recently, \cite{pervez2025riemannian} improved upon this by adopting a randomized subspace of polynomial dimension, yielding an algorithm that is equivalent to \cref{alg-RRSGP}.
Additionally, variants corresponding to the subspace dimension $d=1$ have been investigated in \cite{malvetti2024randomized,magann2023randomized,mcmahon2025equating}.
However, we formalize this procedure within a retraction-based optimization framework. This perspective clarifies that a broad range of optimization techniques can be easily incorporated into \cref{alg-RRSGP} to enhance efficiency (e.g., the exact line-search rule detailed in the next subsection).
Moreover, within this framework, the convergence properties of the algorithm can be established by using standard results from Riemannian optimization theory.

\subsubsection{Exact line-search for RRSGP}\label{sec-exact-lr}

In a quantum implementation of \cref{alg-RRSGP}, running $T$ iterations yields a circuit with $T d$ Pauli-evolution gates, so each additional iteration directly increases circuit depth. Therefore, to keep the circuit shallow, it is crucial to reduce total $T$ by selecting step size $t_k$ carefully at every iteration, rather than using a constant is common in deep learning field.
We now introduce an exact line-search procedure for choosing the step size $t_k$ in Step 3 of \cref{alg-RRSGP}.
Given the projected gradient $\zeta_k$ in \cref{eq-zeta}, the exact line-search is defined as
\begin{align}
  t_k := \arg\min_{t>0}\, \phi(t)
   & = f\!\left(\Trt_{U_k}\!\left(t\zeta_k\right)\right)
  = f\!\left(\prod_{j\in S_k}\exp\!\left(i\omega_k^j P^j t\right)U_k\right) \nonumber \\
   & = \Tr\!\left\{ O\, V(t)\,\psi_k\, V(t)^{\dagger}\right\},
\end{align}
where $V(t):=\prod_{j\in S_k}\exp(i\omega_k^j P^j t)$.
This choice guarantees the maximal decrease of the cost along the projected search direction $\zeta_k$ at each iteration, thereby helping reduce the total number of iterations.
The trade-off is that each iteration now requires (approximately) solving an auxiliary one dimensional subproblem $\min_{t>0} \phi(t)$, which is equivalent to training a single parameter PQC. In our numerical experiments, we use the Adam optimizer~\cite{kingma2014adam} to solve this subproblem, and the results show that it can effectively reduce the iteration count $T$.

\section{Second-order geometry and algorithms}\label{sec-2}

So far, we focused on first-order Riemannian algorithms for solving problem \eqref{pro-manifold}.
A natural question is whether the second-order algorithms can also be implemented on quantum circuits.
It is well known that the Newton method \cite{dennis1996numerical} achieves quadratic convergence when initialized sufficiently close to the optimum.
Hence, applying a Riemannian Newton method \cite{fernandes2017superlinear,absil2009optimization} has the potential to reach the solution in far fewer iterations, thereby producing shallower circuits.
In this section, we discuss how to realize the Riemannian Newton method for solving problem \eqref{pro-manifold}.
The first requirement is the ability to compute the key second-order geometric object, namely, the Riemannian Hessian operator.

\subsection{Second-order geometry}

In the Euclidean setting, given a smooth function $f: \mathbb{R}^n \rightarrow \mathbb{R}$, the Hessian $\nabla^2 f(x)$ at $x \in \mathbb{R}^n$ is an $n \times n$ real symmetric matrix. Equivalently, it can be viewed as a self-adjoint linear operator on $\mathbb{R}^n$.
As a natural extension, the Riemannian Hessian of $f$ at a point on a manifold is a self-adjoint linear operator acting on the tangent space at that point.
The Riemannian Hessian characterizes the second-order variation of the cost function $f$ along tangent directions.
It provides the curvature information to construct second-order optimization algorithms \cite{absil2007trust,boumal2023introduction}, that can accelerate the convergence.
In the next proposition, we present an explicit expression for the Riemannian Hessian of our cost function in problem \eqref{pro-manifold}, whose proof is provided in \ref{app-hess-der}.

\begin{proposition}[Riemannian Hessian]
  \label{prop-hessian}
  For the cost function $f: \mathrm{U}(p) \rightarrow \mathbb{R}$ defined by $f(U)=\Tr \left\{O U \psi_0 U^{\dagger}\right\}$, the Riemannian Hessian of $f$ at $U\in \mathrm{U} (p)$ is the self-adjoint linear operator $\Hess f(U): T_U \to T_U$ given by
  \begin{equation}\label{eq:Hess}
    \Hess f(U)[\Omega U] = \tfrac{1}{2} \left( [O,[\Omega,\psi]] + [[O,\Omega],\psi] \right) U,
  \end{equation}
  where $\psi := U \psi_0 U^\dagger$ and $T_U  = \{ \Omega U : \Omega \in \mathfrak{u}(p)\}$. Identifying $T_U \simeq \mathfrak{u}(p)$ yields the associated operator $\widetilde{\operatorname{Hess}} f(U): \mathfrak{u}(p) \rightarrow \mathfrak{u}(p)$,
  \begin{equation}\label{eq:Hess-tilde}
    \widetilde{\operatorname{Hess}} f(U)[\Omega]=\tfrac{1}{2}([O,[\Omega, \psi]]+[[O, \Omega], \psi]),
  \end{equation}
  which is self-adjoint on the Lie algebra $\mathfrak{u} (p)$. Here, ``self-adjoint'' is understood with respect to the Frobenius inner product $\langle A, B\rangle=\Tr \left (A^{\dagger} B\right) .$
\end{proposition}

Indeed, for any $U\in \mathrm{U} (p)$ and $X \in \mathbb{C}^{p \times p}$, the Riemannian Hessian $\widetilde{\operatorname{Hess}} f(U)$ satisfies
$\widetilde{\operatorname{Hess}} f(U)\left[X^{\dagger}\right]=(\widetilde{\operatorname{Hess}} f(U)[X])^{\dagger}$ and  $\widetilde{\operatorname{Hess}} f(U)[i\, X]=i\, \widetilde{\operatorname{Hess}} f(U)[X]$.
Its linearity and self-adjointness are straightforward to verify. Proofs of these properties can be found in \ref{app-hess-pro}.

\subsection{Second-order algorithms}

In this section, we first introduce a standard Riemannian Newton method and then propose a random subspace variant for the practical implementation.

\subsubsection{Riemannian Newton}

The Riemannian Newton method also fits within the general update framework $U_{k+1}=\mathrm{R}_{U_k}\left(t_k \eta_k\right)$ in \cref{general-update}. The key distinction is that the search direction $\eta_k = \widetilde{\eta} _k U_k \in T_{U_k}$ is now chosen as the Newton direction.

\paragraph{Riemannian Newton equation}
At the iteration $k$, the Newton direction $\widetilde{\eta} _k  = \Omega_k^{\text{N}} \in \mathfrak{u} (p)$ is defined as the solution to the following (Riemannian) Newton equation:
\begin{equation}\label{eq-Newton1}
  \operatorname{Hess} f (U_k)[\Omega_k^{\text{N}} \,U_k] = - \grad f (U_k),
\end{equation}
which constitutes a linear system on the current tangent space $T_{U_k}$. Dropping the rightmost factor $U_k$, the above equation can be equivalently written on the Lie algebra $\mathfrak{u} (p)$ as
\begin{equation}\label{eq-Newton2}
  \widetilde{\operatorname{Hess}} f (U_k)[\Omega_k^{\text{N}}] =- \widetilde{\grad} f(U_k),
\end{equation}
For convenience, let $\psi_k =U_k \psi_0 U_k^{\dagger}$ and define
\begin{equation}
  \mathcal{L}_k(\Omega) \triangleq \widetilde{\operatorname{Hess}} f (U_k)[\Omega]
  =\tfrac{1}{2}\left([O, [\Omega, \psi_k]] +[[O, \Omega], \psi_k]\right).
\end{equation}
Then the Newton direction $\widetilde{\eta} _k = \Omega_k^{\text{N}} \in \mathfrak{u} (p)$ is obtained by solving
\begin{equation}\label{eq-2005}
  \mathcal{L}_k (\Omega_k^{\text{N}})=\left[\psi_k,O\right] .
\end{equation}
The Riemannian Newton update is given by $U_{k+1}=\mathrm{R}_{U_k}(\Omega_k^{\text{N}}U_k)$ with unit step size $t_k\equiv 1$.
We emphasize that the use of a constant unit step size is important for achieving quadratic convergence.

\paragraph{Solving Newton equation}

To solve the Newton equation \cref{eq-2005}, we consider the matrix representation $\mathbf{L}_{k} \in \mathbb{R}^{4^N \times 4^N}$ of the operator $\mathcal{L}_k$ with respect to the orthogonal basis $\{i P^j\}$ of $\mathfrak{u}(p)$. Specifically, for all $r,s=1,\ldots,4^N$, we have
\begin{align}
  (\mathbf{L}_{k})_{rs}
   & = \langle iP^r, \mathcal{L}_k(i P^s)\rangle= \Tr\{ (iP^r)^\dagger \mathcal{L}_k(i P^s)\}= \Tr\{P^r \mathcal{L}_k\left(P^s\right)\}                 \\
   & =\tfrac{1}{2}\Tr\{P^r\left[O, \left[P^s, \psi_k\right]\right]\}+\tfrac{1}{2}\Tr \{P^r\left[\left[O, P^s\right], \psi_k\right]\}                    \\
   & =\tfrac{1}{2}\Tr\{\psi_k\left[\left[P^r,O\right], P^s\right]\}+\tfrac{1}{2}\Tr\{\psi_k\left[\left[P^s, O \right], P^r\right]\}. \label{eq-L-entry}
\end{align}
Under the same basis, let $\mathbf{g}_{k} \in \mathbb{R}^{4^N}$ denote the vector representation of the right-hand side of \cref{eq-2005}, i.e., negative Riemannian gradient.
For $j=1,\ldots,4^N$, its entries are given by
\begin{align}\label{eq-2019}
  \left(\mathbf{g}_k\right)_j
   & = \langle iP^j, -\widetilde{\operatorname{grad}} f\left(U_k\right) \rangle \\
   & = \langle iP^j, \left[\psi_k,O\right] \rangle=\Tr
  \{\left(iP^j\right)^{\dagger} \left[\psi_k, O\right]  \}
  =-i  \Tr \{\psi_k \left[O, P^j\right]\}.
\end{align}
We then classically solve the linear equation
\begin{equation}
  \mathbf{L}_k \mathbf{\Omega}_k^{\text{N}}=\mathbf{g}_k,
\end{equation}
to obtain $\mathbf{\Omega}_k^{\text{N}} \in \mathbb{R}^{4^N}$. The corresponding Newton direction in $\mathfrak{u}(p)$ is recovered as
\begin{equation}
  \Omega_k^{\text{N}} =\sum_{j=1}^{4^N}(\mathbf{\Omega}_k^{\text{N}})_j \cdot iP^j \in \mathfrak{u}(p).
\end{equation}
Finally, by applying the Trotter retraction $\mathrm{R} = \Trt$ in \cref{eq-trt}, the Riemannian Newton update takes the form
\begin{equation}
  U_{k+1}=\operatorname{Trt}_{U_k}(\Omega_k^{\text{N}}U_k)=\left(\prod_{j=1}^{4^N} \exp \left(i (\mathbf{\Omega}_k^{\text{N}})_j P^j \right)\right) U_k.
\end{equation}

\paragraph{Evaluation of coefficients}
The remaining task involves the computation of $\mathbf{g}_{k}$ and $\mathbf{L}_k$.
The evaluation of the entries $\left(\mathbf{g}_k\right)_j$ in \cref{eq-2019} follows the same procedure as that used for computing the gradient coefficients in \cref{alg-RRSGP}.
We consider the following univariate PQC cost function:
\begin{equation}\label{eq-1701}
  g_k^j (x) :=  \Tr \{O U (x)\psi_k U (x)^\dagger \}, \quad U (x) = e^{i x P^j / 2}.
\end{equation}
By \cref{lem-pqc-1}, we obtain
\begin{equation}\label{eq-1708}
  \left(\mathbf{g}_k\right)_j=\left[ g_k^j \left(-\frac{\pi}{2} \right) -  g_k^j \left(\frac{\pi}{2} \right)\right].
\end{equation}

Moreover, we establish the following key lemma, which shows that the matrix representation of the Hessian operator can also be obtained through quantum circuit measurements, in complete analogy with the gradient. The proofs are deferred to \ref{app-proofs}.

\begin{lemma}[Hessian coefficient estimation]\label{lem-vqa-2}
  Given any pure density operator $\psi$, any Hamiltonian $O$, and any Hermitian generators $P$ and $Q$, consider the expectation value on a variational circuit w.r.t. two real parameters $x,y$, i.e.,
  \begin{equation}
    g (x, y)=\Tr\left\{ O W (y) U (x)\psi U (x)^{\dagger} W (y)^{\dagger}\right\},
  \end{equation}
  where $U (x)=e^{i x P / 2}$ and $W (y)=e^{i y Q / 2}$. Then, we obtain the following results concerning the second-order partial derivatives:
  \begin{gather}
    g_{x x}(0,0)
    =\frac{1}{4} \Tr\bigl\{\psi[[P,O], P]\bigr\}, \quad
    g_{x y}(0,0) =\frac{1}{4} \Tr\bigl\{\psi[[Q, O], P]\bigr\}.
  \end{gather}
  Moreover, if $P^2=Q^2=I$, then by the parameter-shift rule \cite[Eq. (11) $\&$ Eq. (13)]{mari2021estimating}, one has
  \begin{align}
    g_{x x}(0,0)
     & =\frac{1}{2}\left[g\left(\frac{\pi}{2},0\right)+g\left(-\frac{\pi}{2},0\right)-2 g (0, 0)\right], \label{eq-1711}                                                                                                            \\
    g_{x y}(0,0)
     & =\frac{1}{4}\left[g\left(\frac{\pi}{2}, \frac{\pi}{2}\right)-g\left(\frac{\pi}{2}, -\frac{\pi}{2}\right)-g\left(-\frac{\pi}{2}, \frac{\pi}{2}\right)+g\left(-\frac{\pi}{2}, -\frac{\pi}{2}\right)\right] .   \label{eq-1712}
  \end{align}
  Note that if we aim to compute $\Tr\bigl\{ \psi [[P,O ], Q] \bigr\}$, it suffices to swap the order of the parameterized unitaries $U (x)$ and $W (y)$ in circuits.

\end{lemma}

To estimate the coefficients $(\mathbf{L}_{k})_{rs}$ in \cref{eq-L-entry}, we consider the following bivariate PQC cost functions:
\begin{align}
   & g_{k}^{rs} (x,y) := \Tr\left\{ O W (y) U (x)\psi_k  U (x)^{\dagger} W (y)^{\dagger}\right\}, \\
   & \qquad U (x) = e^{i x P^s / 2}, \quad W (y) = e^{i y P^r / 2},
\end{align}
and $g_{k}^{sr} (x,y) $, defined analogously by exchanging the roles of $P^r$ and $P^s$. Since both $P^r$, $P^s$ are Pauli words, applying \cref{eq-1712} of \cref{lem-vqa-2} yields
\begin{align}
  (\mathbf{L}_{k})_{rs}
       & =\tfrac{1}{2}\Tr\left\{\psi_k\left[\left[P^r, O \right], P^s\right]\right\}+\tfrac{1}{2}\Tr\left\{\psi_k\left[\left[ P^s,O\right], P^r\right]\right\}                                                                                                             \\
       & = \tfrac{1}{2}\left[g_{k}^{rs}\left(\frac{\pi}{2}, \frac{\pi}{2}\right)-g_{k}^{rs}\left(\frac{\pi}{2}, -\frac{\pi}{2}\right)-g_{k}^{rs}\left(-\frac{\pi}{2}, \frac{\pi}{2}\right)+g_{k}^{rs}\left(-\frac{\pi}{2}, -\frac{\pi}{2}\right)\right] \notag             \\
  \, + & \, \tfrac{1}{2}\left[g_{k}^{sr}\left(\frac{\pi}{2}, \frac{\pi}{2}\right)-g_{k}^{sr}\left(\frac{\pi}{2}, -\frac{\pi}{2}\right)-g_{k}^{sr}\left(-\frac{\pi}{2}, \frac{\pi}{2}\right)+g_{k}^{sr}\left(-\frac{\pi}{2}, -\frac{\pi}{2}\right)\right].  \label{eq-2032}
\end{align}
When $P^r$ and $P^s$ commute (which occurs for half of all Pauli-word pairs), we have $g_{k}^{r s}(x, y) \equiv g_{s r}^{(k)}(x, y),$ and the expression above simplifies to
\begin{equation}\label{eq-21023}
  \left(\mathbf{L}_k\right)_{r s}=\left[g_{k}^{r s}\left(\frac{\pi}{2}, \frac{\pi}{2}\right)-g_{k}^{r s}\left(\frac{\pi}{2},-\frac{\pi}{2}\right)-g_{k}^{r s}\left(-\frac{\pi}{2}, \frac{\pi}{2}\right)+g_{k}^{r s}\left(-\frac{\pi}{2},-\frac{\pi}{2}\right)\right] .
\end{equation}

In particular, for the diagonal entries $(\mathbf{L}_{k})_{jj}
  = \Tr\{\psi_k[[P^j,O], P^j]\}$, we consider the same PQC function $g_{k}^{j} (x) $ defined in \cref{eq-1701}. Applying \cref{eq-1711} of \cref{lem-vqa-2} yields
\begin{equation}
  (\mathbf{L}_{k})_{jj}
  =2 \left[g_{k}^{j} \left(\frac{\pi}{2}\right)  +g_{k}^{j} \left(-\frac{\pi}{2}\right)-g_{k}^{j} \left(0\right)\right].
\end{equation}
Note that $g_{k}^{j}(0)$ corresponds to the current cost value $f(U_k)$ and is independent of $P^j$; moreover, the values $g_{k}^{j}\left(\pm \frac{\pi}{2}\right)$ have already been computed in \cref{eq-1708}. Therefore, the calculation of diagonal components does not require any additional costs.
At this point, we have completed all the components required to implement the standard Riemannian Newton method on quantum circuits.

\subsubsection{Riemannian Random Subspace Newton (RRSN)}

From the perspective of quantum circuit implementation, Riemannian gradient descent and Riemannian Newton differ only in how they compute the coefficients (i.e., angles) of the gates. In terms of convergence rate, the gradient descent method is at best linearly convergent, whereas the Newton method enjoys the well-known quadratic convergence, meaning it can often reach highly accurate solutions in far fewer iterations.

However, the standard Newton method faces a major computational bottleneck in terms of quantum circuit evaluation cost. For an $N$-qubit system, forming the full Hessian matrix $\mathbf{L}_k \in \mathbb{R}^{4^N \times 4^N}$ requires on the order of $O (16^N )$ evaluations. To overcome it, we propose the Riemannian Random Subspace Newton (RRSN) method. Its core idea is essentially the same as that of the previously introduced RRSGP in \cref{sec-rrsgp}: at each iteration, we randomly select only $d=\operatorname{poly} (N)$ Pauli words out of the full set of $4^N$, while keeping the rest of the procedure unchanged. The RRSN algorithm is summarized in \cref{alg-RRSN}.
Under this random subspace restriction, the Newton equation
\begin{equation}\label{eq-565989}
  \mathbf{L}_k \boldsymbol{\Omega}_k^{\mathrm{N}}=\mathbf{g}_k \in \R^d
\end{equation}
reduces to a $d \times d$ linear system. Uniformly sampling $d$ Pauli basis considerably lowers the problem dimension, yet still preserves convergence in expectation, following the same rationale as RRSGP.

In \cref{alg-RRSN}, we additionally incorporate two standard engineering techniques for Newton methods: (Step 4) Hessian regularization via a positive definite modification and (Step 5) an Armijo backtracking line-search.

\paragraph{Hessian regularization}

In the preceding discussion, we implicitly assumed that the Newton equation \cref{eq-565989} admits a solution. However, this need not hold in general. Therefore, we introduce a regularization term $\delta_k \mathbf{I}_d$ to ensure that the Hessian matrix $\mathbf{L}_k \in \mathbb{R}^{d \times d}$ becomes positive definite, and hence invertible. Specifically, we set
\begin{equation}
  \delta_k=\max \left\{0, \rho-\lambda_{\min }\left (\mathbf{L}_k\right)\right\},
\end{equation}
for a small $\rho>0$, and $\lambda_{\text {min }}\left (\mathbf{L}_k\right)$ denotes the smallest eigenvalue of $\mathbf{L}_k$. The modified Hessian is then
\begin{equation}
  \widetilde{\mathbf{L}}_k: =\mathbf{L}_k+\delta_k \mathbf{I}_d \succeq \rho \mathbf{I}_d \succ 0.
\end{equation}
Consequently, the (regularized) Newton equation $\widetilde{\mathbf{L}}_k \boldsymbol{\Omega}_k^{\mathrm{N}}=\mathbf{g}_k$ admits a unique solution $\boldsymbol{\Omega}_k^{\mathrm{N}} = \widetilde{\mathbf{L}}_k^{-1} \mathbf{g}_k \in \mathbb{R}^d$. The corresponding Newton direction in the chosen random subspace is then recovered as
\begin{equation}\label{eq-4496}
  \Omega_k^{\mathrm{N}}: =\sum_{j \in S_k}  (\boldsymbol \Omega_k^{\mathrm{N}} )_j \cdot i P^j \in \operatorname{span}_{\mathbb{R}}\left\{i P^j: j \in S_k\right\} \subseteq \mathfrak{u} (p).
\end{equation}
Moreover, in this case, Newton direction $\Omega_k^{\mathrm{N}} U_k \in T_{U_k}$ is a descent direction, i.e., it forms an obtuse angle with the Riemannian gradient in the tangent space $T_{U_k}$. Indeed, by \cref{eq-4496,eq-2019},
\begin{align}
   & \left\langle\operatorname{grad} f\left(U_k\right), \Omega_k^{\mathrm{N}}U_k \right\rangle_{U_k}
  =\left\langle \widetilde{\operatorname{grad}} f(U_k), \Omega_k^{\mathrm{N}}\right\rangle
  =\sum_{j \in S_k}  (\boldsymbol \Omega_k^{\mathrm{N}} )_j\left\langle \widetilde{\operatorname{grad}} f(U_k),   i P^j\right\rangle                                                                                                        \\
   & =-\sum_{j\in S_k} (\boldsymbol{\Omega}_k^{\mathrm{N}} )_j \left(\mathbf{g}_k\right)_j =-\mathbf{g}_k^\top \mathbf{\Omega}_k^{\text{N}}= -\mathbf{g}_k^\top \widetilde{\mathbf{L}}_k^{-1} \mathbf{g}_k <0, \label{eq-descent-direction}
\end{align}
where the last inequality holds whenever $\mathbf{g}_k \neq 0$.

\begin{remark}
  If $\mathbf{g}_k \approx 0$, i.e., $\left (\mathbf{g}_k\right)_j=\langle i P^j, -\widetilde{\operatorname{grad}} f\left (U_k\right) \rangle \approx 0$, then the components of the Riemannian gradient along the sampled subspace are nearly zero. This situation can arise for two distinct reasons. First, the iterate may already be sufficiently close to an optimum, in which case the full Riemannian gradient itself is close to zero and the algorithm should terminate. Second, the randomly chosen subspace may be uninformative, so that the projection nearly annihilates the gradient even though the full gradient remains non-negligible; in this case, the subspace should be resampled. This phenomenon is especially pronounced when $d$ is very small (e.g., $d=1$), since the subspace then captures only a limited amount of gradient information.
\end{remark}

\paragraph{Armijo backtracking line-search}

Finally, we perform Armijo backtracking along the Newton direction. We first take the full step $t_k \leftarrow 1$ and check whether the following \textit{Armijo condition} holds:
\begin{equation}\label{eq-armijo-cond}
  f\left (\operatorname{Trt}_{U_k}\left (t_k \Omega_k^{\mathrm{N}} U_k\right)\right) \leq f\left (U_k\right)+c \, t_k\left\langle\operatorname{grad} f\left (U_k\right), \Omega_k^{\mathrm{N}} U_k\right\rangle .
\end{equation}
If not, we shrink $t_k \leftarrow \beta t_k$ (typically $\beta=0.5, c=10^{-4}$) and repeat the test until the above is satisfied. Since $\Omega_k^N U_k$ is a descent direction shown in \cref{eq-descent-direction}, this procedure terminates after finitely many loops; see \cite[Lemma 3.1]{nocedal1999numerical}.

In general, the Newton method enjoys quadratic convergence only when the initial iterate $U_0$ is sufficiently close to an optimum $U^*$; with a poor initialization, the standard Newton method may fail to converge. To improve robustness, we enforce the Armijo condition \cref{eq-armijo-cond} (also called \textit{sufficient decrease} condition), which promotes the convergence from an arbitrary initial $U_0$.
Numerical experiments indicate that the Armijo condition \cref{eq-armijo-cond} is often already satisfied at $t_k=1$, so this loop does not consume excessive quantum resources in practice. This technique is merely an optional safeguard.

\begin{algorithm}[H]
  \caption{Riemannian Random Subspace Newton (RRSN) Method for \eqref{pro-manifold}}
  \label{alg-RRSN}
  \begin{algorithmic}[1]
    \Require
    $N$-qubits system, Hamiltonian $O$, initial state $\psi_0$,  subspace dimension $d = \poly(N)$. Pauli set $\mathcal{P}^N=\{P^j\}_{j=1}^{4^N}$ with $P^j=\bigotimes_{\ell=1}^N \sigma_{\ell}^j$, $\sigma_{\ell}^j\in\{I,X,Y,Z\}$. Let $f_0=\Tr\{O\psi_0\}.$ Set $ \rho = 10^{-1}$, $c=10^{-4}$, $\beta=0.5$.

    \For{$k = 0, 1, \ldots$}
    \State \textbf{Step 1: Subspace Sampling}
    \State Uniformly sample index set $S_k = \{j_1, \dots, j_d\} \subseteq \{1, \dots, 4^N\}$.

    \vspace{0.4em}
    \State \textbf{Step 2: Gradient Estimation $\mathbf{g}_k \in \mathbb{R}^d$} \Comment{$-\widetilde{\grad} f\left(U_k\right)$}
    \For{$\alpha = 1, \dots, d$}
    \State Let $j \leftarrow j_\alpha$ and define $g_k^{j}(x)=\Tr\{Oe^{i x P^{j} / 2}\psi_k e^{-i x P^{j} / 2}\}$.
    \State Set $(\mathbf{g}_k)_\alpha = [ g_k^j(-\frac{\pi}{2}) - g_k^j(\frac{\pi}{2}) ]$.  \Comment{By parameter-shift.}
    \EndFor

    \vspace{0.4em}
    \State \textbf{Step 3: Hessian Estimation $\mathbf{L}_k\in \mathbb{R}^{d \times d}$} \Comment{$\mathcal{L}_k$}
    \For{$\alpha = 1, \dots, d$}
    \For{$\beta = \alpha, \dots, d$}
    \State Let $r \leftarrow j_\alpha$ and $s \leftarrow j_\beta$.
    \If{$\alpha = \beta$} \Comment{Diagonal: reuse values from Line 7.}
    \State Set $(\mathbf{L}_k)_{\alpha \alpha} = 2 \left[ g_k^r(\frac{\pi}{2}) + g_k^r(-\frac{\pi}{2}) - 2f_k \right]$.
    \Else \Comment{Off-diagonal via parameter-shift.}
    \State Define $g_k^{rs}(x,y)=\Tr\{O e^{i y P^{r} / 2}e^{i x P^{s} / 2} \psi_k   e^{-i x P^{s} / 2} e^{-i y P^{r} / 2}\}$.
    \State Let $\Delta(g) := g(\frac{\pi}{2}, \frac{\pi}{2}) - g(\frac{\pi}{2}, -\frac{\pi}{2}) - g(-\frac{\pi}{2}, \frac{\pi}{2}) + g(-\frac{\pi}{2}, -\frac{\pi}{2})$.
    \If{$[P^r, P^s] = 0$}  \Comment{50\% probability.}
    \State Set $(\mathbf{L}_k)_{\alpha\beta} = \Delta(g_k^{rs})$.
    \Else
    \State Set $(\mathbf{L}_k)_{\alpha\beta} =\tfrac{1}{2} [\Delta(g_k^{rs}) + \Delta(g_k^{sr})]$.
    \EndIf
    \EndIf
    \State Set $(\mathbf{L}_k)_{\beta\alpha} \leftarrow (\mathbf{L}_k)_{\alpha\beta}$.
    \EndFor
    \EndFor

    \vspace{0.4em}
    \State \textbf{Step 4: Hessian Regularization \& Solve Newton equation}
    \State Set $\delta_k \leftarrow \max\{0, \rho - \lambda_{\min}(\mathbf{L}_k)\}$. \Comment{Positive definite modification.}
    \State Solve $(\mathbf{L}_k + \delta_k \mathbf{I}_d) \,\mathbf{\Omega}_k^{\text{N}} = \mathbf{g}_k$ for Newton direction $\mathbf{\Omega}_k^{\text{N}} \in \mathbb{R}^d$.

    \vspace{0.4em}
    \State \textbf{Step 5: Armijo Backtracking Line-Search (Optional)}
    \For{$t_k \in \{1, \beta, \beta^2, \ldots\}$}
    \State Form trial state $|\psi_{\text{trial}}\rangle \leftarrow  \prod_{\alpha=1}^d \exp(i (\mathbf{\Omega}_k^{\text{N}})_\alpha P^{j_\alpha}t_k )  |\psi_{k}\rangle$.
    \State Set $f_{\text{new}} \leftarrow \Tr \{O \psi_{\text{trial}}\}$.
    \If{$f_{\text{new}} \leq f_k - c\, t_k \,\mathbf{g}_k^\top \mathbf{\Omega}_k^{\text{N}}$}
    \State \textbf{break}
    \EndIf
    \EndFor
    \State Set $\psi_{k+1} \leftarrow \psi_{\text{trial}}$, $f_{k+1} \leftarrow f_{\text{new}}$, $k \leftarrow k+1$.
    \EndFor
  \end{algorithmic}
\end{algorithm}

\subsubsection{RRSN when subspace dimension \texorpdfstring{$d=1$}{d=1}}\label{sec-resource}

The circuit measurement per-iteration cost of RRSGP (\cref{alg-RRSGP}) scales as $O(d)$, whereas that of RRSN (\cref{alg-RRSN}) scales as $O(d^2)$.
RRSGP is often used with subspace dimension $d=1$, a setting already adopted in prior works \cite{malvetti2024randomized,magann2023randomized,mcmahon2025equating}.
In this case, each RRSGP iteration uniformly samples a single index $j_k\in\{1,\ldots,4^N\}$ and defines
\begin{equation}
  g_k(x)=\operatorname{Tr}\!\left\{O\, e^{i x P^{j_k}/2}\, \psi_k\, e^{-i x P^{j_k}/2}\right\}.
\end{equation}
The Riemannian gradient is then projected onto the one-dimensional direction $iP^{j_k}$ using two function evaluations $g_k\!\left(\pm \frac{\pi}{2}\right)$.
However, when RRSN is also configured with $d=1$ (see \cref{alg-RRSN-d1} for intuition), the required circuit resources reduce to the same two evaluations $g_k\!\left(\pm \frac{\pi}{2}\right)$.

\begin{algorithm}[H]
  \caption{RRSN of \cref{alg-RRSN} with subspace dimension $d=1$ for \eqref{pro-manifold}}
  \label{alg-RRSN-d1}
  \begin{algorithmic}[1]
    \Require $N$-qubits system, Hamiltonian $O$, initial state $\psi_0$. Pauli set $\mathcal{P}^N=\{P^j\}_{j=1}^{4^N}$.
    Set $f_0=\Tr\{O\psi_0\}$, $\rho=10^{-1}$, $c=10^{-4}$, $\beta=0.5$.

    \For{$k = 0,1,\ldots$}
    \State \textbf{(1) Sample one direction:} uniformly sample $j_k\in\{1,\ldots,4^N\}$.
    \State Define $g_k(x)=\Tr\{O e^{i x P^{j_k}/2}\psi_k e^{-i x P^{j_k}/2}\}$.

    \State \textbf{(2) Gradient (scalar):} $\mathbf{g}_k \leftarrow g_k(-\tfrac{\pi}{2}) - g_k(\tfrac{\pi}{2})$.
    \State \textbf{(3) Hessian (scalar):} $\mathbf{L}_k \leftarrow 2 [g_k(\tfrac{\pi}{2}) + g_k(-\tfrac{\pi}{2}) - 2 f_k ]$.

    \State \textbf{(4) Regularize \& Newton step:}
    $\delta_k \leftarrow \max\{0,\rho - \mathbf{L}_k\}$, $\omega_k \leftarrow \frac{\mathbf{g}_k}{\mathbf{L}_k+\delta_k}$.
    \Statex \hspace{\algorithmicindent} Equivalently, $\omega_k \leftarrow \frac{\mathbf{g}_k}{\max\{\mathbf{L}_k,\rho\}}$.

    \State \textbf{(5) Armijo backtracking (Optional):}
    \For{$t \in \{1,\beta,\beta^2,\ldots\}$}
    \State $|\psi_{\text{trial}}\rangle \leftarrow \exp\!\big(i\,\omega_k P^{j_k} t\big)\,|\psi_k\rangle$.
    \State $f_{\text{new}} \leftarrow \Tr\{O\psi_{\text{trial}}\}$.
    \If{$f_{\text{new}} \le f_k - c\, t\, g_k\,\omega_k$}
    \State \textbf{break}
    \EndIf
    \EndFor
    \State Set $\psi_{k+1} \leftarrow \psi_{\text{trial}}$, $f_{k+1} \leftarrow f_{\text{new}}$, $k \leftarrow k+1$.
    \EndFor
  \end{algorithmic}
\end{algorithm}

For $d=1$, RSSGP performs a one-dimensional gradient step along a random direction, while RRSN performs a one-dimensional Newton step along the same random direction. The difference is a scalar curvature correction, namely, whether one divides by the curvature (see (3) in \cref{alg-RRSN-d1})
\begin{equation}
  \mathbf{L}_k: =\left\langle i P^{j_k}, \mathcal{L}_k\left (i P^{j_k}\right)\right\rangle=2\left[g_k\left (\frac{\pi}{2}\right)+g_k\left (-\frac{\pi}{2}\right)-2 f_k\right].
\end{equation}
In summary, when $d=1$, the first-order method RSSGP can be upgraded to the second-order method RRSN at no additional cost. In the numerical experiment of \cref{sec-comparsion-d1} and \Cref{fig:combined_exp_6}, we observe that RRSN is more efficient than RSSGP when $d=1$.

\section{Numerical experiments}\label{sec-experiments}

In this section, we present numerical experiments to evaluate our proposed Riemannian algorithms, RRSGP (\cref{alg-RRSGP}) and RRSN (\cref{alg-RRSN}), for model \eqref{pro-manifold}, and compare them with the VQA algorithm for model \eqref{pro-pqc}.
Quantum circuits are implemented using PennyLane \cite{bergholm2018pennylane}.
We consider the Heisenberg XXZ Hamiltonian for an $N$-qubit system, defined as
\begin{equation}
  O = \sum_{i=1}^{N} X_i X_{i+1} + \sum_{i=1}^{N} Y_i Y_{i+1} + \Delta \sum_{i=1}^{N} Z_i Z_{i+1},
\end{equation}
where $\Delta = 0.5$ and periodic boundary conditions are imposed, e.g., $X_{N+1} \equiv X_1$.
Let $f_{\mathrm{ground}}$ denote the true ground state energy of $O$.
As performance metrics, we report the \textit{Energy error} $\lvert f_k - f_{\mathrm{ground}} \rvert$ (where $f_k$ denotes the cost value at the $k$-th iteration) and the (Riemannian) \textit{Gradient norm} $\|[\psi_k, O]\|_{F}$. The algorithm terminates when either of the following criteria is met: (i) the gradient norm falls below $10^{-9}$; or (ii) the relative change in energy satisfies $\lvert f_{k+1}-f_k\rvert / \lvert f_k\rvert < 10^{-10}$.
The VQA employs a two-layer Hardware-Efficient Ansatz (HEA), as shown in \Cref{fig:vqa_circuit}; the Adam optimizer is used with a learning rate of 0.01, and all parameters are initialized to zero.

\begin{figure}
  \centering
  \includegraphics[width=1\linewidth]{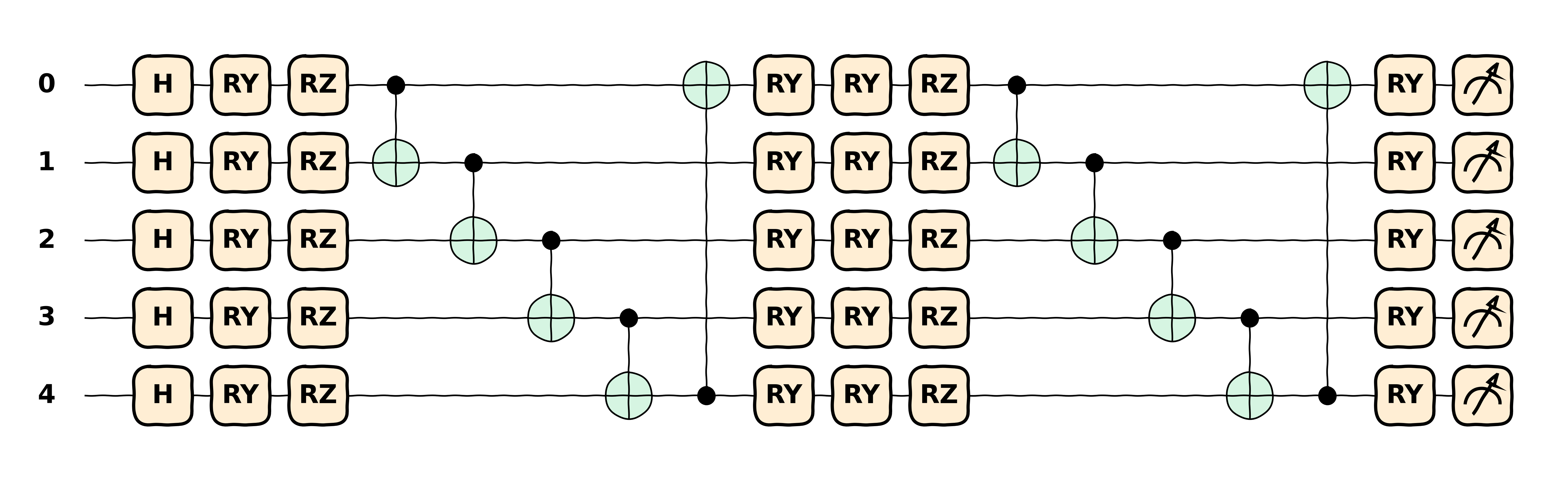}
  \caption{The circuit of the VQA algorithm used in our experiments. It consists of a two-layer hardware-efficient ansatz, where each $R_Y$ and $R_Z$ gate is parameterized by a single free parameter.}
  \label{fig:vqa_circuit}
\end{figure}

\subsection{Setting I: Idealized setting (without random subspaces)}

We first consider an idealized setting in which no random subspace is used, so all $4^N$ Pauli words are selected. In this case, RRSGP reduces to the standard Riemannian gradient descent, while RRSN corresponds to the standard Riemannian Newton method. Under this setting, the algorithms have full access to complete information about both the Riemannian gradients and Hessians.

\subsubsection{Comparison of typical convergence rates}

\Cref{fig:combined_exp_1} illustrates typical convergence behaviors of four algorithms: RRSGP (fixed), which uses a fixed step size of 0.1; RRSGP (exact), which uses an exact line-search discussed in \cref{sec-exact-lr} (implemented via the PennyLane built-in Adam optimizer with a maximum of 30 iterations, learning rate 0.1); RRSN; and VQA. All methods start from the uniform superposition state.

As seen in (a) of \Cref{fig:combined_exp_1}, RRSGP with exact line-search exhibits faster convergence than the fixed-step version. Although the curves of RRSGP (exact) and RRSN appear to overlap in (a), their differences become apparent on the logarithmic scale in (b): RRSN shows a typical quadratic convergence rate, whereas RRSGPs (regardless of the step-size strategy) exhibit only linear convergence. This distinction is further confirmed by the gradient norm behavior in (c) and (d).

Note that RRSN is implemented as a modified Newton method: a regularization term is added to ensure that the Hessian is positive definite, and an Armijo backtracking line-search is employed. As a result, the energy error decreases monotonically. In contrast, the norm of the Riemannian gradient is not necessarily monotone, as shown in panel (c), which is consistent with typical optimization behavior.

For the VQA baseline, even after 500 iterations, the final energy error remains 1.14, indicating that it fails to converge to the true solution. This may be due to the limited expressivity of the parametrized quantum state \cite{sim2019expressibility,holmes2022connecting}, or to optimization difficulties such as the barren plateau phenomenon \cite{mcclean2018bp,larocca2025barren}.

\subsubsection{VQA warm starts accelerate Riemannian algorithms}

As Riemannian algorithms, the circuit depth of both RRSGPs and RRSN increases monotonically with the number of iterations. To mitigate this issue, we adopt a two-stage strategy: we first run a low-cost VQA on a fixed shallow circuit to quickly bring the quantum state close to the optimum, and then switch to a Riemannian algorithm for further refinement. In particular, if the VQA warm start places the state within the quadratic convergence region, RRSN can converge in only a few iterations.

Here, we consider the 5-qubit XXZ Hamiltonian and take the state obtained after 200 VQA iterations as the initial state. \Cref{fig:combined_exp_2} compares RRSGP and RRSN with and without this VQA warm start. RRSGP uses a fixed step size of 0.1.
As shown in panel (a), RRSGP without a VQA warm start (i.e., RRSGP (uniform-init)) is attracted to a saddle point during the iterations, corresponding to an intermediate energy level of the Hamiltonian. Although it eventually escapes, this stagnation wastes computational resources.
By contrast, the curves of RRSGP (VQA-init) show that the VQA warm start effectively steers RRSGP away from this saddle point.
Meanwhile, RRSN is more robust: even from a cold start, it can bypass the saddle point directly.

Panels (b) and (d) of \Cref{fig:combined_exp_2} further show that warm started RRSN allows it to enter the quadratic convergence region earlier, thereby reducing the total circuit depth required to reach the target accuracy. In summary, for our Riemannian algorithms, using VQA as a warm start is always beneficial.

\subsection{Setting II: Practical setting (with random subspaces)}

In the previous subsection, the experiments were conducted under an idealized setting, where the random subspace dimension $d=4^N$. We now turn to the practical case with $d<4^N$. We continue to study the $N=4$ qubits XXZ Hamiltonian. Our goal is to investigate how the choice of the subspace dimension $d$ affects the performance of three algorithms: RSSGP (fixed) with a fixed step size of 0.1, RSSGP (exact) with exact line-search, and RRSN. Motivated by the preceding experiments, we employ a two-layer VQA warm start for all methods to improve their efficiency.

\subsubsection{RRSN is more robust to the dimension of random subspaces}

In \Cref{fig:combined_vqa_exp_4}, we set the subspace dimension to $d=1 \ (0.4 \%), d=4 \ (1.6 \%), d=8 \ (3.1 \%), d=16 \ (6.2 \%), d=32 \ (12.5 \%), d=64 \ (25.0 \%), d=128 \ (50.0 \%)$, and $d=256 \ (100.0 \%)$.
The numbers in parentheses report $d / 4^N$ (with $4^N=256$), i.e., the fraction of full gradient/Hessian information accessed per iteration. For each choice of $d$, we conduct 20 independent runs and plot the averaged curves.

In \Cref{fig:combined_vqa_exp_4}, we observe that for all algorithms the convergence rate gradually deteriorates as $d$ decreases, indicating a strong dependence on the subspace dimension.
The benefit, however, is a significant reduction in per-iteration resource cost, since the circuit measurement complexity of RSSGP scales as $O (d)$, whereas that of RRSN scales as $O (d^2)$.
Moreover, from the linear scale plots in panels (a), (c), and (e), using a small $d$ (even $d=1$) does not appear drastically worse in the early iterations. However, in the later stage it becomes difficult to attain high-precision solution.

Panel (f) illustrates that, as $d$ decreases from the full dimension 256 to 1, the convergence behavior of RRSN degrades from quadratic convergence to superlinear convergence, and eventually to linear convergence.
Interestingly, when $d \geq 64$, RRSN achieves performance that is nearly identical to the full case $d=256$, while the per-iteration circuit measurement cost at $d=64$ is only $0.0625=64^2 / 256^2$ of that at $d=256$. Moreover, for $d=32$, where RRSN still exhibits superlinear convergence, the per-iteration measurement cost is merely $0.015625=32^2 / 256^2$ of the full case.
By contrast, for RSSGPs, reducing $d$ from 256 to 32 leads to a significant degradation in convergence performance.
Hence, RRSN is more robust than RRSGP to the dimension of random subspaces.

\subsubsection{Comparison of Riemannian algorithms under low dimensional subspaces}

In \Cref{fig:combined_vqa_exp_3}, we compare RSSGP (fixed), RSSGP (exact), and RRSN under four subspace dimensions: $d=4 \ (1.6 \%), d=16 \ (6.2 \%), d=64 \ (25.0 \%)$, and $d=256 \ (100.0 \%)$.
The solid curves show the average over 20 independent runs, while the shaded bands indicate the $10 \%-90 \%$ quantile range, capturing the variability across runs.

In \Cref{fig:combined_vqa_exp_3}, we observe that as $d$ increases, the randomness diminishes and their trajectories eventually stabilize to deterministic curves.
For all values of $d$, RSSGP (exact) consistently outperforms RSSGP (fixed). When $d=4$ or $d=16$, RRSN degrades to linear convergence; nevertheless, it remains faster than both RSSGP variants, albeit with larger fluctuations across runs.

\subsubsection{Comparison of Riemannian algorithms under one dimensional subspaces}\label{sec-comparsion-d1}

We next investigate the extreme case $d=1$. As discussed in \cref{sec-resource}, in this regime both RSSGPs and RRSN update along a single random direction; the only difference is that RRSN additionally incorporates curvature information into the coefficient, without any extra cost. Although RRSN no longer enjoys quadratic convergence when $d=1$, we still expect it to outperform RSSGPs, which is indeed observed in practice.

We consider the XXZ Hamiltonian with the number of qubits $N \in \{2, 3, 4, 5\}$. Fixing $d=1$, we run RSSGP (fixed), RSSGP (exact), and RRSN (\cref{alg-RRSN-d1}) for 10 independent runs and report the average total number of iterations (for ease of illustration, the energy error tolerance is set to $10^{-5}$) in \Cref{fig:combined_exp_6}. Since $d=1$, the circuit depth is equal to the iteration count. \Cref{fig:combined_exp_6} shows how the iteration counts of the three algorithms scale with $N$. It can be seen that RRSN is significantly faster than RSSGPs.

\begin{figure}
  \centering
  \includegraphics[width=1.0\linewidth]{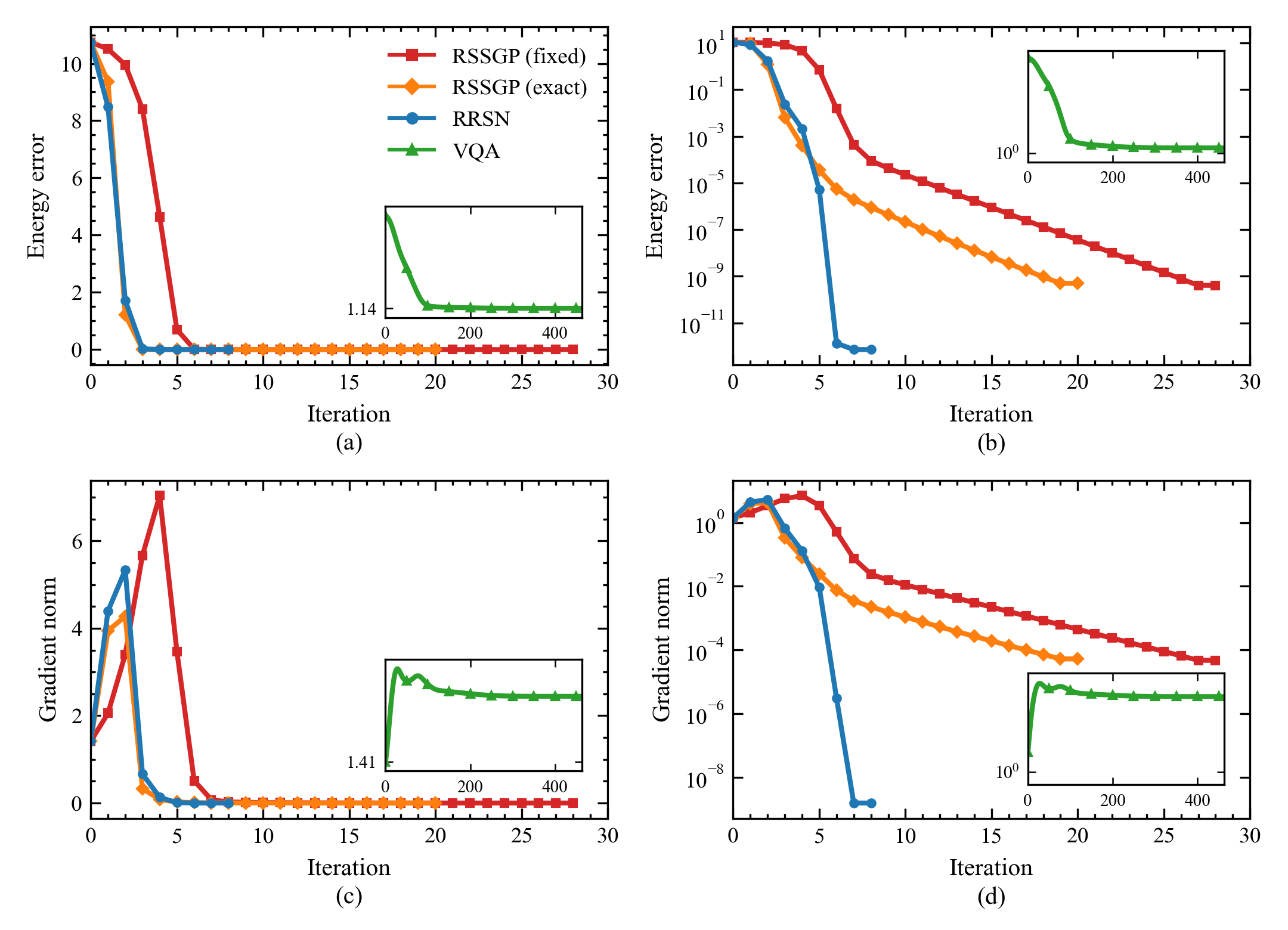}
  \caption{Comparison of typical convergence rate without considering random subspaces. Experiments are conducted on a 4-qubit XXZ, starting from the uniform state. Panels (a) and (b) show the energy error versus iteration number on linear and logarithmic scales, respectively, while (c) and (d) display the evolution of the gradient norm. Although exact line-search accelerates the initial convergence of RRSGP, only RRSN exhibits quadratic convergence, whereas RRSGPs achieve at best linear convergence.}
  \label{fig:combined_exp_1}
\end{figure}

\begin{figure}
  \centering
  \includegraphics[width=1.0\linewidth]{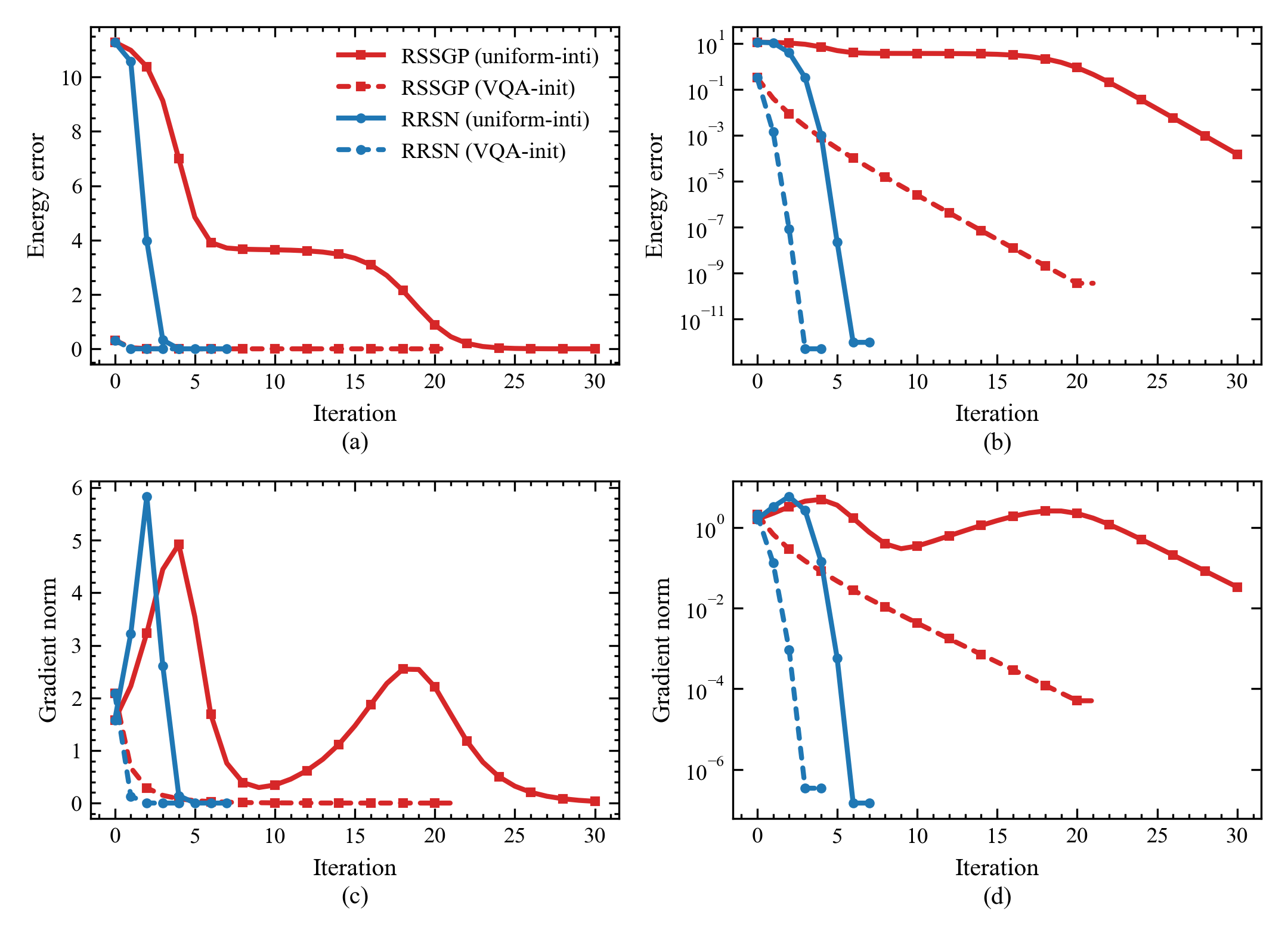}
  \caption{Impact of VQA warm start on performance for RRSGP and RRSN. Experiments are performed on a 5-qubit XXZ Hamiltonian. Solid curves correspond to cold starts from the uniform state, while dashed curves indicate VQA warm starts. Panels (a) and (c) show that cold start RRSGP is prone to being trapped near a saddle point, whereas the VQA warm start avoids this issue. RRSN, by contrast, bypasses the saddle point even with a cold start. Panels (b) and (d) further demonstrate that VQA warm start enables RRSN to enter the quadratic convergence region earlier.}
  \label{fig:combined_exp_2}
\end{figure}

\begin{figure}
  \centering
  \includegraphics[width=1.0\linewidth]{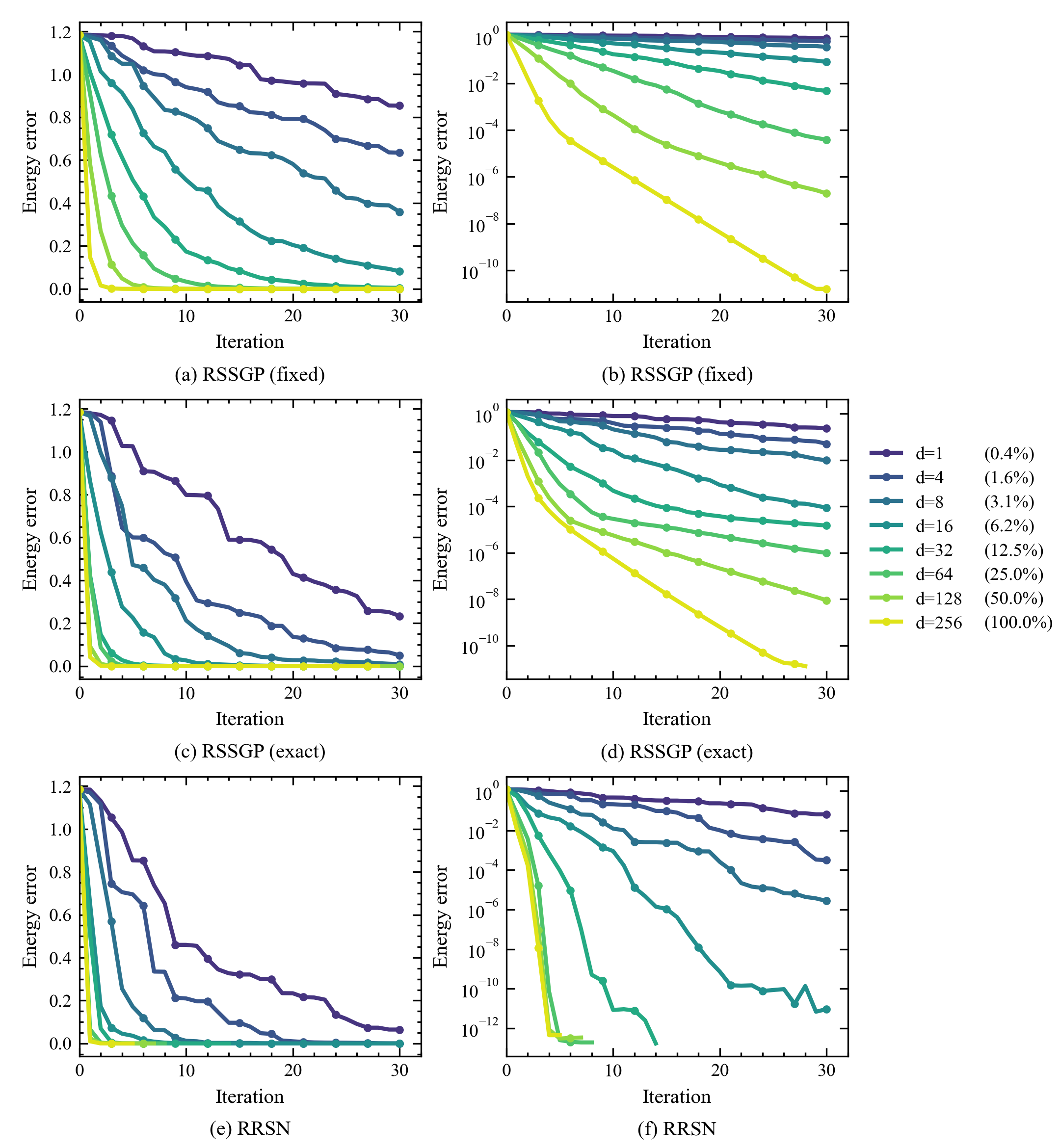}
  \caption{Convergence behavior of RSSGP (fixed), RSSGP (exact), and RRSN under practical random subspaces ($d<4^N$). Each curve is averaged over 20 independent runs. Panels (a) and (b) show the energy error of RSSGP (fixed) on linear and logarithmic scales, respectively; panels (c) and (d) show RSSGP (exact); and panels (e) and (f) show RRSN. As $d$ decreases, all methods slow down, while RRSN exhibits a gradual transition from quadratic to superlinear and then linear convergence; for $d=64$ or $128$, it achieves performance close to the full dimensional case. RRSN is more robust than RRSGP to the dimension of random subspaces.}
  \label{fig:combined_vqa_exp_4}
\end{figure}

\begin{figure}
  \centering
  \includegraphics[width=1.0\linewidth]{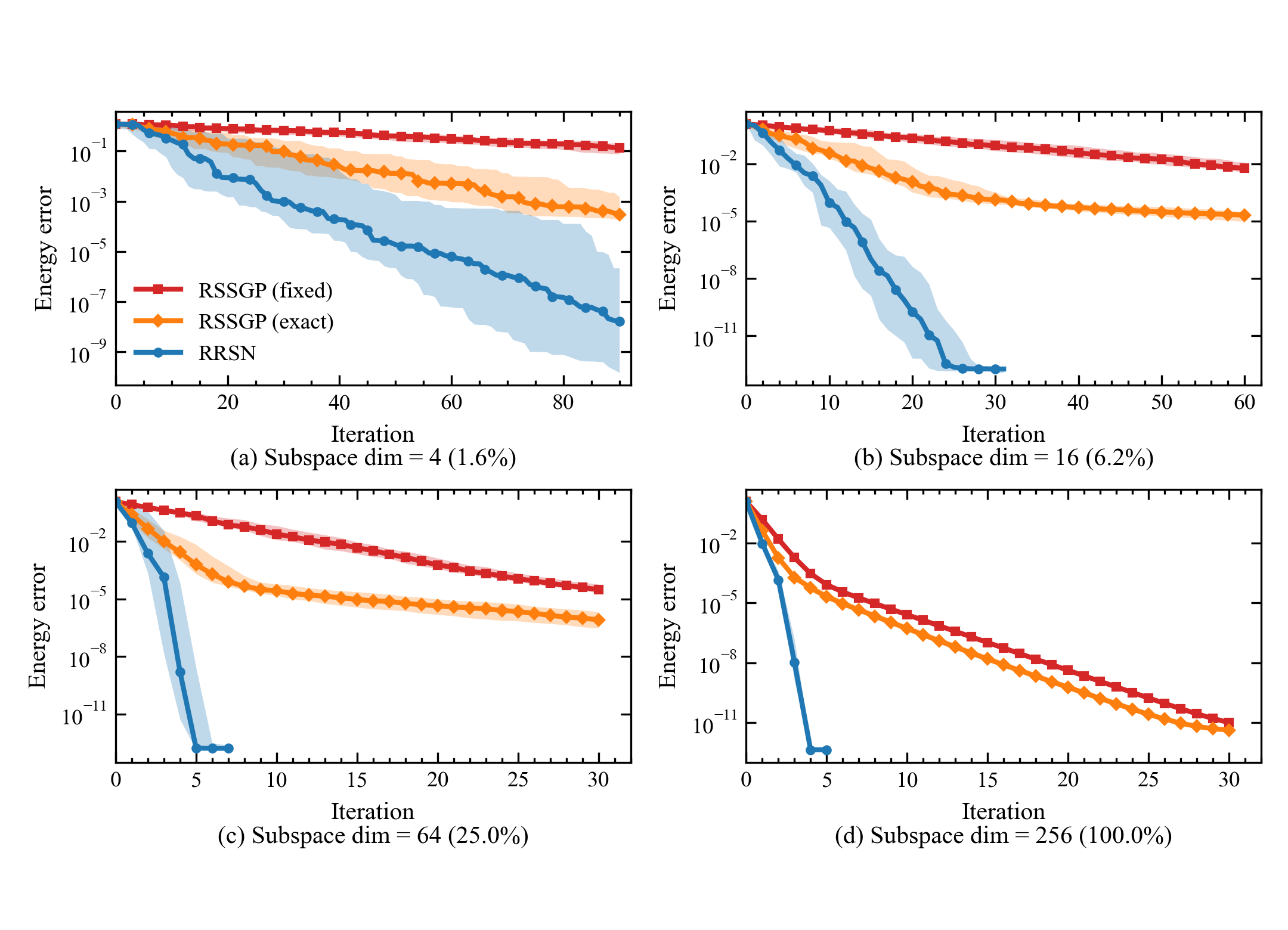}
  \caption{Comparison of RSSGP (fixed), RSSGP (exact), and RRSN under practical random subspaces ($d<4^N$). Curves show the average over 20 runs, and the shaded bands indicate the $10 \%-90 \%$ quantile ranges. As $d$ increases, the variability across runs decreases. RSSGP (exact) always outperforms RSSGP (fixed); RRSN remains faster overall even at small $d$.}
  \label{fig:combined_vqa_exp_3}
\end{figure}

\begin{figure}
  \centering
  \includegraphics[width=0.6\linewidth]{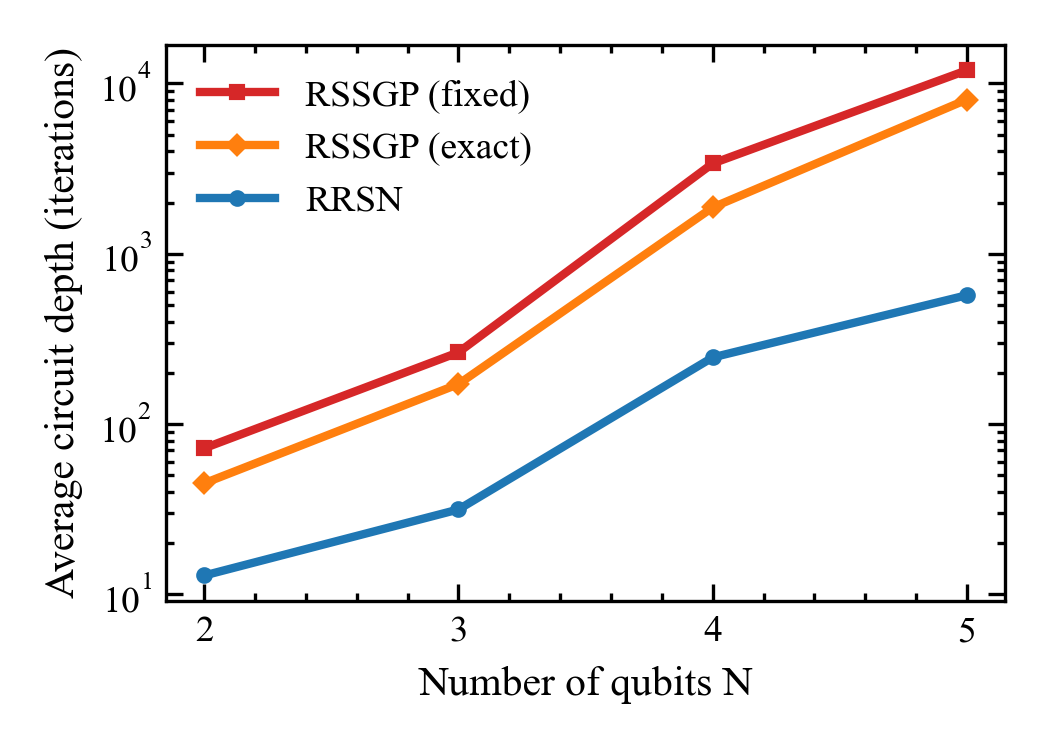}
  \caption{Comparison of RSSGP (fixed), RSSGP (exact), and RRSN under one-dimensional random subspaces ($d=1$). For the XXZ Hamiltonian, we consider qubit numbers $N=2,3,4,5$. Each point reports the average total iteration count over 10 runs required to reach an energy error below $10^{-5}$. The results show that although the per-iteration cost of the three algorithms is identical when $d=1$, RRSN consistently requires fewer iterations than RRSGPs.}
  \label{fig:combined_exp_6}
\end{figure}

\newpage

\section{Conclusion}\label{sec-discussion}

In this work, we present a systematic geometric approach to quantum circuit design by optimizing the energy cost directly over the unitary group $\mathrm{U}(p)$.
We develop a comprehensive \textit{retraction-based} Riemannian optimization framework, in which every algorithmic component (from the retraction map to the evaluation of geometric objects) is fully implementable on quantum hardware.
This framework bridges quantum circuit optimization with the mature field of Riemannian optimization \cite{absil2009optimization,boumal2023introduction}, enabling us to use established convergence theory and advanced algorithmic techniques in a rigorous manner for quantum computing.

By interpreting the Trotter approximation as a valid retraction, we provided a unified theoretical foundation for existing randomized RGD techniques \cite{pervez2025riemannian,malvetti2024randomized,magann2023randomized,mcmahon2025equating}, formalized here as the Riemannian Random Subspace Gradient Projection (RRSGP) method.
More importantly, we extended this framework to second-order algorithms.
We derived the explicit Riemannian Hessian for the energy cost and demonstrated, for the first time, that it can be estimated via parameter-shift rules using only quantum measurements.
Based on this, we proposed the Riemannian Random Subspace Newton (RRSN) method.

Our numerical experiments highlight the significant advantages of incorporating second-order information.
RRSN achieves a quadratic convergence rate, reaching high-precision solutions in substantially fewer iterations than first-order methods and standard VQA baselines.
On the other hand, RRSN exhibits strong robustness to the dimension of the random subspace.
Even in the extreme case of a one-dimensional subspace ($d=1$), where RRSN and RRSGP share the same per-iteration measurement cost, RRSN consistently outperforms the first-order approach by using curvature information to rescale the step size.

To address the challenges of the NISQ era, specifically circuit depth and measurement costs, we validated a hybrid strategy: using a shallow VQA as a ``warm start'' effectively avoids potential saddle points, or positions the system directly within the quadratic convergence region of the RRSN.
This approach strikes a balance between the hardware efficiency of fixed ansatzes and the high-accuracy refinement capabilities of Riemannian optimizations.

The adoption of a retraction-based Riemannian optimization framework for \eqref{pro-manifold} establishes a systematic foundation for applying a diverse suite of Riemannian algorithms to quantum circuit design.
While this work focuses on quantum-implementable Riemannian gradient descent and Riemannian Newton methods, the framework facilitates the potential extension to other advanced techniques, including Riemannian conjugate gradient \cite{abrudan2009conjugate,sato2022riemannian}, Riemannian trust-region \cite{absil2007trust,boumal2023introduction}, Riemannian quasi-Newton (e.g., BFGS) \cite{huang2015broyden,huang2018riemannian}, and Riemannian accelerated or adaptive schemes (e.g., Adam) \cite{zhang2018towards,becigneul2018riemannian}.
We expect that such geometric perspectives will play an increasingly important role in unlocking the potential of quantum algorithms.
In fact, in another work \cite{lai2025grover} by the authors, we were the first to obtain the quadratic speedup of Grover's quantum search algorithm using standard convergence analysis tools of Riemannian gradient descent.

\section*{Data availability}

The data and codes used in this study are available from public repositories\footnote{https://github.com/GALVINLAI/RiemannianQCD}.

\section*{Acknowledgments}
This work was supported by the National Natural Science Foundation of China under the grant numbers 12501419, 12288101 and 12331010, and the National Key R\&D Program of China under the grant number 2024YFA1012901.
DA acknowledges the support by the Quantum Science and Technology-National Science and Technology Major Project via Project 2024ZD0301900, and the Fundamental Research Funds for the Central Universities, Peking University.

\appendix

\section{Proofs}\label{app-proofs}

\begin{proof}[Proof of \cref{lem-pqc-1}]
  Recall that $g(x)=\Tr \left\{\psi U(x)^{\dagger}O U(x) \right\}$ with $U(x)=e^{i x P / 2}$.
  Since $\dot{U} (x) = \frac{d}{dx} U (x) = \frac{i}{2} P U (x)$ and $\dot{U} (x)^\dagger  = -\frac{i}{2} U (x)^\dagger P,$ we obtain
  \begin{equation}\label{eq-2014}
    \frac{d}{d x}\left(U(x)^{\dagger} O U(x)\right)=\dot{U}(x)^{\dagger} O U(x)+U(x)^{\dagger} O \dot{U}(x)=\frac{i}{2} U(x)^{\dagger}[O, P] U(x).
  \end{equation}
  Substituting back into the trace yields
  \begin{equation}
    g^{\prime}(x)=\Tr \left\{\psi \frac{d}{d x}\left(U(x)^{\dagger} O U(x)\right)\right\}=\frac{i}{2} \Tr \left\{\psi U(x)^{\dagger}[O, P] U(x)\right\}.
  \end{equation}
  In particular, evaluating at $x=0$ yields the desired result. We complete the proof.
\end{proof}

\begin{proof}[Proof of \cref{lem-vqa-2}]
  Recall that
  \begin{equation}
    g (x, y)=\Tr \left\{\psi U (x)^{\dagger} W (y)^{\dagger} O W (y) U (x)\right\}
  \end{equation}
  with $U(x)=e^{i x P / 2}$ and $W(y)=e^{i y Q / 2}$. For convenience, write $W = W (y)$ and $\dot W = \frac{d}{d y} W (y)$, and likewise $U = U (x)$ and $\dot U = \frac{d}{dx} U (x) $.
  We first differentiate $g (x, y)$ with respect to $y$, and then differentiate the resulting expression with respect to $x$.
  Since $\dot W=  \frac{i}{2} Q W$ and $\dot W^\dagger =-  \frac{i}{2}W^\dagger Q$, the same calculation as in \cref{eq-2014} yields
  \begin{equation}
    \frac{d}{d y}\left(W (y)^{\dagger} O W (y)\right)
    =\frac{i}{2} W (y)^{\dagger}[O, Q] W (y) .
  \end{equation}
  Hence,
  \begin{align}
    g_y(x,y)
     &
    =\Tr\Bigl\{\psi  U  ^\dagger
    \frac{d}{d y }\Bigl (W (y)^\dagger  O  W (y)\Bigr)  U  \Bigr\}                              \\
     & =\frac{i}{2}  \Tr\bigl\{\psi  U (x)^\dagger  W (y)^\dagger[O, Q]  W (y)  U (x)\bigr\}  .
  \end{align}
  Next, since $\dot U =  \frac{i}{2} P U$ and $\dot U^\dagger =-  \frac{i}{2} U^\dagger P$, the same computation as in \cref{eq-2014} shows that
  \begin{equation}
    \frac{d}{d x }\Bigl (U (x)^\dagger Z U (x)\Bigr)
    =\frac{i}{2} U (x)^\dagger[Z, P] U (x)
  \end{equation}
  for any constant $Z$ that does not depend on $x$. Therefore,
  \begin{align}
    g_{x y}(x,y)
     & =\frac{d  }{d  x   }  g_y (x,y)
    =\frac{d}{d x }\left[\frac{i}{2} \Tr\left\{\psi U(x)^{\dagger} W^{\dagger}[O, Q] W U(x)\right\}\right]            \\
     & =\frac{i}{2} \Tr\left\{\psi \frac{d}{d x }\left(U^{\dagger}  \left(W^{\dagger}[O, Q] W\right) U\right)\right\} \\
     & =\frac{i}{2} \Tr\left\{\psi\left(\frac{i}{2} U^{\dagger}\left[W^{\dagger}[O, Q] W, P\right] U\right)\right\}   \\
     & =-\frac{1}{4} \Tr\left\{\psi U (x)^\dagger\bigl[ W (y)^\dagger[O, Q] W (y), P\bigr] U (x)\right\},
  \end{align}
  where in the third line we treat $W^{\dagger}[O, Q] W$ as the constant operator $Z$ with respect to $x$. Evaluating at $(0,0)$, where $U(0)=W(0)=I$, yields
  \begin{equation}
    g_{x y}(0,0)
    =\frac{1}{4} \Tr\bigl\{\psi[[Q,O], P]\bigr\}.
  \end{equation}
  The remaining results can be obtained using analogous approaches. We complete the proof.
\end{proof}

\section{Details about Riemannian Hessian}\label{app-hess}

Extending the Hessian of a real-valued function defined on a Euclidean space to a real-valued function defined on a manifold gives rise to the Riemannian Hessian. For detailed introductions to the Riemannian Hessian, we refer the reader to textbooks \cite{absil2009optimization,boumal2023introduction}. Here, we directly adopt existing methodology and derive the Riemannian Hessian of our cost function $f$ in \eqref{pro-manifold} explicitly. We then discuss its basic properties.

\subsection{Derivation of Riemannian Hessian}\label{app-hess-der}

Recall that a linear operator $L: V \rightarrow V$ on an inner product space $(V,\langle\cdot, \cdot\rangle)$ is self-adjoint if $\langle L x, y\rangle=\langle x, L y\rangle$ for all $x, y \in V$.

\begin{proposition}[Riemannian Hessian]\label{prop-hessian-app}
  For the cost function $f: \mathrm{U}(p) \rightarrow \mathbb{R}$ defined by $f(U)=\Tr \left\{O U \psi_0 U^{\dagger}\right\}$, the Riemannian Hessian of $f$ at $U\in \mathrm{U} (p)$ is the self-adjoint linear operator $\Hess f(U): T_U \to T_U$ given by
  \begin{equation}\label{eq:Hess-app}
    \Hess f(U)[\Omega U] = \tfrac{1}{2} \left( [O,[\Omega,\psi]] + [[O,\Omega],\psi] \right) U,
  \end{equation}
  where $\psi := U \psi_0 U^\dagger$ and $T_U  = \{ \Omega U : \Omega \in \mathfrak{u}(p)\}$. Identifying $T_U \simeq \mathfrak{u}(p)$ yields the associated operator $\widetilde{\operatorname{Hess}} f(U): \mathfrak{u}(p) \rightarrow \mathfrak{u}(p)$,
  \begin{equation}\label{eq:Hess-tilde-app}
    \widetilde{\operatorname{Hess}} f(U)[\Omega]=\tfrac{1}{2}([O,[\Omega, \psi]]+[[O, \Omega], \psi]),
  \end{equation}
  which is self-adjoint on the Lie algebra $\mathfrak{u} (p)$. Here, ``self-adjoint'' is understood with respect to the Frobenius inner product $\langle A, B\rangle=\Tr \left (A^{\dagger} B\right) .$
\end{proposition}

\begin{proof}[Proof of \cref{prop-hessian}]
  Here, we provide a detailed derivation, following \cite[Corollary 5.16]{boumal2023introduction}, by first differentiating a smooth extension of the Riemannian gradient  (see \cref{eq-grad}) $U \mapsto \grad f(U)=[O, U \psi_0 U^\dagger]U $ in the ambient space $\mathbb{C}^{p \times p} \supseteq \mathrm{U}(p)$ and then orthogonally projecting the result onto $T_U$.
  Specifically, consider the extension function $G: \mathbb{C}^{p \times p} \to \mathbb{C}^{p \times p}$ of $\grad f$ defined by
  \begin{equation}
    G (U) := A (U) U, \, \text{ with }  \,  A (U): =[ O, U \psi_0 U^\dagger].
  \end{equation}
  Using the usual differential $D A (U)[V]= \lim _{t \rightarrow 0} \tfrac{A(U+t V)-A(U)}{t}=\bigl[  O, V  \psi_0  U^\dagger   +  U  \psi_0  V^\dagger  \bigr]$ and applying the product rule, the differential of $G$ at $U$ in any direction $V\in \mathbb{C}^{p \times p}$ is given by
  \begin{align}
    D G (U)[V]
     & =\bigl (D A (U)[V]\bigr) U + A (U) V                         \\
     & = \bigl[ O, V \psi_0 U^\dagger + U \psi_0 V^\dagger \bigr] U
    + \bigl[ O, U \psi_0 U^\dagger \bigr] V.
  \end{align}
  Let $\psi: = U \psi_0 U^\dagger$ and restrict the direction $V$ to the subspace $T_U = \{ \Omega U: \Omega \in \mathfrak{u} (p) \}$. Then, we obtain $D G (U)[\Omega U] =[O, [\Omega, \psi]] U +[O, \psi] \Omega U \triangleq Z.$
  Next, projecting $Z$ onto $T_U$ via \cref{eq-projection} yields the Riemannian Hessian:
  \begin{equation}\label{eq-1609}
    \operatorname{Hess} f (U)[\Omega U]
    := \Skew (Z U^{\dagger}) U = \Skew  \{ \underbrace{[O, [\Omega, \psi]]}_{ (1)} + \underbrace{[O, \psi] \Omega}_{ (2)}  \} U.
  \end{equation}
  Note that term (1) is skew-Hermitian, while term (2) is a product of two skew-Hermitian matrices; thus, by \cref{eq-iden-skew}, \cref{eq-1609} reduces to
  \begin{align}
    \operatorname{Hess} f (U)[\Omega U] & = ([O, [\Omega, \psi]] + \frac{1}{2}[[O, \psi], \Omega] ) U           \\
                                        & =\frac{1}{2}\left([O, [\Omega, \psi]] +[[O, \Omega], \psi] \right) U,
  \end{align}
  where the last equality follows from the Jacobi identity $[O, [\Omega, \psi]]+[\Omega, [\psi, O]]+[\psi, [O, \Omega]]=0$.
  Finally, it is straightforward to verify that $\operatorname{Hess} f(U)$ is a self-adjoint linear map from $T_U$ to $T_U$; see \ref{app-hess-pro} for details.
\end{proof}

\subsection{Properties of Riemannian Hessian}\label{app-hess-pro}

In this subsection, we summarize basic properties of the Riemannian Hessian associated with our cost function. For convenience, we denote by $\mathcal{L}$ the tilde-form Hessian in \cref{eq:Hess-tilde-app} and, without loss of generality, drop the prefactor $\tfrac{1}{2}$.
Accordingly, we consider the linear map $\mathcal{L}:V \to V$ defined by
\begin{equation}\label{eq-linear-L}
  \mathcal{L} (X) =[ O, [ X, \psi]] + [[ O, X], \psi ],
\end{equation}
where $O$ is a given Hermitian matrix and $\psi=|\psi\rangle \langle \psi|$ is a density operator.
The vector space $V$ may be taken as the full space $\mathbb{C}^{p\times p}$, the Lie algebra $\mathfrak{u}(p)=\left\{\Omega \in \mathbb{C}^{p \times p}: \Omega^{\dagger}=-\Omega\right\}$, or the space of Hermitian matrices $\mathcal{H}(p)=\left\{H \in \mathbb{C}^{p \times p}: H^{\dagger}=H\right\}$. This is justified because $\mathcal{L}$ preserves the adjoint structure of its argument; see the next proposition.

\begin{proposition}[Adjoint-preservation]\label{prop-ad-preserve}
  Consider linear operator $\mathcal{L}$ in \cref{eq-linear-L}. For all $X \in \mathbb{C}^{p \times p}$, one has $\mathcal{L}(X^{\dagger})=\mathcal{L}(X)^{\dagger}$. Consequently,
  \begin{enumerate}
    \item $\mathcal{L}(X) \in \mathcal{H}(p)$ whenever $X \in \mathcal{H}(p)$;
    \item  $\mathcal{L}(X) \in \mathfrak{u}(p)$ whenever $X \in \mathfrak{u}(p)$.
  \end{enumerate}
\end{proposition}
\begin{proof}
  We use the identities $[A,B]^{\dagger} = [B^{\dagger},A^{\dagger}]$, $O^{\dagger}=O$, and $\psi^{\dagger}=\psi$. Then, $\mathcal{L}(X^\dagger)
    = [O,[X^\dagger,\psi]] + [[O,X^\dagger],\psi]
    = [[\psi,X^\dagger],O] + [\psi,[X^\dagger,O]]
    = [[X,\psi]^\dagger,O] + [\psi,[O,X]^\dagger]
    = \mathcal{L}(X)^\dagger.$
\end{proof}

Next, we show the linearity of $\mathcal{L}$. While $\mathcal{L}$ is defined over complex matrices, its behavior over real subspaces is of particular interest for our Riemannian optimization, where the tangent space $T_U$ of the unitary group is a real vector space. The following result is easy to prove.

\begin{proposition}[Linearity]
  Consider linear operator $\mathcal{L}$ in \cref{eq-linear-L}.
  \begin{enumerate}
    \item $\mathcal{L}$ is $\C$-linear over complex vector space $\C^{p\times p}$. In particular, $\mathcal{L}(i P)=i \mathcal{L}(P)$ hold for any input $P$.
    \item $\mathcal{L}$ is $\R$-linear over real vector space $\C^{p\times p} \simeq \R^{2p\times 2p}$. In particular, $\mathcal{L}$ is $\mathbb{R}$-linear over real vector subspace $\mathfrak{u}(p)$ or $\mathcal{H}(p)$.
  \end{enumerate}
\end{proposition}

An important characteristic of the Riemannian Hessian is self-adjointness (or symmetry). Recall that a linear operator $L: V \rightarrow V$ on an inner product space $(V,\langle\cdot, \cdot\rangle)$ is self-adjoint if $\langle L x, y\rangle=\langle x, L y\rangle$ for all $x, y \in V$. We now show that $\mathcal{L}$ satisfies this condition under several natural inner product structures. The proof is straightforward and is therefore omitted.

\begin{proposition}[Self-adjointness]
  Consider linear operator $\mathcal{L}$ in \cref{eq-linear-L}. $\mathcal{L}$ is self-adjoint under the following vector spaces:
  \begin{enumerate}
    \item On the complex vector space $\mathbb{C}^{p\times p}$ equipped with the Frobenius inner product $\langle A, B\rangle=\Tr(A^\dagger B)$.
    \item On the real vector space $\mathbb{C}^{p\times p} \simeq \mathbb{R}^{2p\times 2p}$ equipped with the real inner product $\langle A, B\rangle=\Re \Tr(A^\dagger B)$.
    \item On the real vector space $\mathcal{H}(p)$ equipped with $\langle A, B\rangle=\Re \Tr(A^\dagger B)=\Tr(AB)$.
    \item On the real vector space $\mathfrak{u}(p)$ equipped with $\langle A, B\rangle=\Re \Tr(A^\dagger B)=-\Tr(AB)$.
  \end{enumerate}
\end{proposition}

Beyond self-adjointness, a stronger property for analyzing convergence is the positive semidefiniteness of the Riemannian Hessian. The Hessian is said to be positive semidefinite if $\langle A, \operatorname{Hess} f(U)[A]\rangle \geq 0$ for all $A \in T_U$. While global positive definiteness (for all $U$) would imply that the function is \textit{geodesically convex} (a generalization of Euclidean convexity to manifolds), this is impossible in our setting. It is a known topological result that on a compact manifold, such as the unitary group, every geodesically convex function must be constant \cite[Corollary 11.10]{boumal2023introduction}. Hence, the cost function in problem \eqref{pro-manifold} cannot be geodesically convex, and its Hessian is not positive semidefinite for all $U$. Instead, positive semidefiniteness holds only within a neighborhood of a local minimizer \cite[Proposition 6.3]{boumal2023introduction}.

\bibliographystyle{plainurl}
\bibliography{mybib}

\end{document}